\documentclass[twoside, 11pt]{article}
\usepackage{jmlr2e}
\usepackage{times}
\usepackage{graphicx}
\usepackage{hyperref}
\usepackage{amsfonts}
\usepackage{amsmath}
\usepackage{subfigure}
\usepackage{natbib}

\usepackage{algorithm}
\usepackage{algorithmic}

\usepackage{times}

\newtheorem{claim}[theorem]{Claim}

\newcommand{\dist}{\mathrm{dist}}
\newcommand{\eat}[1]{}

\begin{document}

\title{Local algorithms for interactive clustering}

\author{\name Pranjal Awasthi \email pawasthi@cs.cmu.edu\\
        \addr Department of Computer Science\\
        Princeton University
        \AND
        \name Maria Florina Balcan \email ninamf@cs.cmu.edu\\
        \addr School of Computer Science\\
        Carnegie Mellon University
        \AND
        Konstantin Voevodski \email kvodski@google.com\\
        \addr Google, NY, USA
}
        
\editor{}
\maketitle        

%\twocolumn[
%\icmltitle{Local algorithms for interactive clustering}

% unless you've provided the [accepted] option to the icml2014
%\icmlauthor{Pranjal Awasthi}{pawasthi@cs.cmu.edu}
%\icmladdress{Carnegie Mellon University,
 %           Pittsburgh, USA}
%\icmlauthor{Maria-Florina Balcan}{ninamf@cs.cmu.edu}
%\icmladdress{Carnegie Mellon
%University, School of Computer Science, Pittsburgh, USA}
%\icmlauthor{Konstantin Voevodski}{kvodski@google.com}
%\icmladdress{Google,
 %           NY, USA}

%\icmlkeywords{boring formatting information, machine learning, ICML}

%\vskip 0.3in
%]

\begin{abstract}
We study the design of interactive clustering algorithms
for data sets satisfying natural stability
assumptions. Our algorithms start with any initial clustering
and only make local changes in each step; both are desirable features in many applications.
We show that in this constrained setting one can still design provably efficient algorithms that produce
accurate clusterings.
We also show that our algorithms perform well on real-world data.
\end{abstract}

%\maketitle
\section{Introduction}

Clustering is usually studied in an unsupervised learning scenario where the goal is to partition the data given pairwise similarity information. Designing provably-good clustering algorithms is challenging because given a similarity function there may be multiple plausible clusterings of the data. Traditional approaches resolve this ambiguity by making assumptions on the data-generation process. For example, there is a large body of work that focuses on clustering data that is generated by a mixture of  Gaussians~\cite{AM05, KSV05, Sanjoy99, AK01, BV08, KalaiMV10, MoitraV10, BelkinS10}. Although this helps define the ``right'' clustering one should be looking for, real-world data rarely comes from such well-behaved probabilistic models.  An alternative approach is to use limited user supervision to help the algorithm reach the desired answer. This approach has been facilitated by the availability of cheap crowd-sourcing tools in recent years. In certain applications such as search and document classification, where users are willing to help a clustering algorithm arrive at their own desired
answer with a small amount of additional prodding, interactive algorithms are very useful.  Hence, the study of interactive clustering algorithms has become an exciting new area of research.

In many practical settings we already start with a fairly good
clustering computed with semi-automated techniques. For example, consider
an online news portal that maintains a large collection of news articles.
The news articles are clustered on the ``back-end,'' and are used to serve several
``front-end'' applications such as recommendations and article profiles. For such a system, we do not have the freedom to compute arbitrary clusterings and present them to the user, which has been proposed in prior work.  But it is still feasible to get limited feedback and \emph{locally} edit the clustering. In particular, we may only want to change the ``bad'' portion revealed by the feedback without changing the rest of the clustering.  Motivated by these observations, in this paper we study the problem of designing local algorithms for interactive clustering.

We propose a theoretical interactive model and provide strong experimental evidence supporting the practical applicability our algorithms. In our model we start with an initial clustering of the data. The algorithm then interacts with the user in stages.  In each stage the user provides limited feedback on the current clustering in the form of {\em split} and {\em merge} requests.  The algorithm then makes a {\em local} edit to the clustering that is consistent with user feedback.  Such edits are aimed at improving the problematic part of the clustering pointed out by the user.  The goal of the algorithm is to quickly converge (using as few requests as possible) to a clustering that the user is happy with - we call this clustering the target clustering.

In our model the user may request a certain cluster to be {\em split} if it is overclustered (intersects two or more clusters in the target clustering). The user may also request to {\em merge} two given clusters if they are underclustered (both intersect the same target cluster).  Note that the user may not tell the algorithm how to perform the split or the merge; such input is infeasible because it requires a manual analysis of all the objects in the corresponding clusters.  We also restrict the algorithm to only make {\em local} changes at each step, i.e., in response we may change only the cluster assignments of the points in the corresponding clusters. If the user requests to split a cluster $C_i$, we may change only the cluster assignments of points in $C_i$, and if the user requests to merge $C_i$ and $C_j$ , we may only reassign the points in $C_i$ and $C_j$.

The split and merge requests described above are a natural form of feedback.  It is easy for users to spot over/underclustering issues and request the corresponding splits/merges (without having to provide any additional information about how to perform the edit). For our model to be practically applicable, we also need to account for noise in the user requests. In particular, if the user requests a merge, only a fraction or a constant number of the points in the two clusters may belong to the same target cluster. Our model~(See Section~\ref{sec:notation}) allows for such noisy user responses.

We study the complexity of algorithms in the above model (the number of edits requests needed to find the target clustering) as a function of the error of the initial clustering. The initial error may be evaluated in terms of {\em underclustering} error $\delta_u$  and {\em overclustering} error $\delta_o$~(See Section~\ref{sec:notation}).  Because the initial error may be fairly small,\footnote{Given 2 different $k$ clusterings, $\delta_u$ and $\delta_o$ is atmost $k^2$.} we would like to develop algorithms whose complexity depends polynomially on $\delta_u$, $\delta_o$ and only logarithmically on $n$, the number of data points.  We show that this is indeed possible given that the target clustering satisfies a natural \emph{stability} property~(see Section~\ref{sec:notation}).  We also develop algorithms for the well-known correlation-clustering objective function~\cite{Bansal04}, which considers pairs of points that are clustered inconsistently with respect to the target clustering~(See Section~\ref{sec:notation}).

As a pre-processing step, our algorithms compute the average-linkage tree of all the points in the data set.  Note that if the target clustering $C^{\ast}$ satisfies our \emph{stability} assumption, then the average-linkage tree must be consistent with $C^{\ast}$ (see Section~\ref{sec:eta-merge}).  However, in practice this average-linkage tree is much too large to be directly interpreted by the users.  Still, given that the edit requests are somewhat consistent with $C^{\ast}$, we can use this tree to efficiently compute local edits that are consistent with the target clustering.  Our analysis then shows that after a limited number of edit requests we must converge to the target clustering.

\noindent \textbf{Our Results}\\
In Section~\ref{sec:eta-merge} we study the $\eta$-merge model. Here we assume that the user may request to split a cluster $C_{i}$ only if $C_{i}$ contains points from several ground-truth clusters.  The user may request to
merge $C_{i}$ and $C_{j}$ only if an $\eta$-fraction of points in each $C_{i}$ and $C_{j}$ are from the same ground-truth cluster.

For this model for $\eta > 0.5$, given an initial clustering with overclustering error $\delta_o$ and underclustering error $\delta_u$, we present an algorithm that requires $\delta_o$ split requests and $2(\delta_u + k) \log_{\frac 1 {1-\eta}} n$ merge requests to find the target clustering, where $n$ is the number of points in the dataset.
For $\eta > 2/3$, given an initial clustering with correlation-clustering error $\delta_{cc}$, we present an algorithm that requires at most $\delta_{cc}$ edit requests to find the target clustering.

In Section~\ref{sec:unrestricted-merge} we relax the condition on the merges and allow the user to request a merge even if $C_i$ and $C_j$ only have a single point from the same target cluster.
We call this the {\em unrestricted-merge} model. Here the requirement on the accuracy of the user response is much weaker and we need to make further assumptions on the nature of the requests. More specifically, we assume that each merge request is chosen uniformly at random from the set of feasible merges.
Under this assumption we present an algorithm that with probability at least $1-\epsilon$ requires $\delta_o$ split requests and $O(\log \frac k {\epsilon} {{\delta}^2_u})$ merge requests to find the target clustering.

We develop several algorithms for performing the split and merge requests under different assumptions.  Each algorithm uses the global average-linkage tree $T_{glob}$ to compute a local clustering edit.  Our splitting procedure finds the node in $T_{glob}$ where the corresponding points are first split in two.  It is more challenging to develop a correct merge procedure, given that we allow ``impure'' merges, where one or both clusters have points from another target cluster (other than the one that they both intersect).  To perform such merges, in the $\eta$-merge model we develop a procedure to extract the ``pure'' subsets of the two clusters, which must only contain points from the same target cluster.  Our procedure searches for the deepest node in $T_{glob}$ that has enough points from both clusters.  In the unrestricted-merge model, we develop another merge procedure that either merges the two clusters or merges them and splits them.  This algorithm always makes progress if the proposed merge is ``impure,'' and makes progress on average if it is ``pure'' (both clusters are subset of the same target cluster).

When the data satisfies stronger assumptions, we present more-scalable split and merge algorithms that do not require any global information.  These procedures compute the edit by only considering the points in the user request and the similarities between them.

In Section~\ref{sec:experiments} we demonstrate the effectiveness of our algorithms on real data.  We show that for the purposes of splitting known over-clusters, the splitting procedure proposed here computes the best splits, when compared to other well-known techniques.
We also test the entire proposed framework on newsgroup documents data, which is quite challenging for traditional unsupervised clustering methods~\cite{Telgarsky12, HellerG05, Dasgupta08, Dai10, Boulis04, Zhong05}.  Still, we find that our algorithms perform fairly well; for larger settings of $\eta$ we are able find the target clustering after a limited number of edit requests.

\noindent \textbf{Related work}\\
Interactive models for clustering studied in previous works~\cite{BalcanB08, AwasthiZ10} were inspired by an analogous model for learning under feedback~\cite{angluin}. In this model, the algorithm can propose a hypothesis to the user~(in this case, a clustering of the data) and get some feedback regarding the correctness of the current hypothesis. As in our model, the feedback considered is split and merge queries. The goal is to design efficient algorithms which use very few queries to the user. 
A critical limitation in prior work is that the algorithm has the freedom to choose any arbitrary clustering as the starting point and can make arbitrary changes at each step. Hence these algorithms may propose a series of ``bad'' clusterings to the user to quickly prune the search space and reach the target clustering.
Our interactive clustering model is in the context of an initial
clustering; we are restricted to only making local changes
to this clustering to correct the errors pointed out by
the user. This model is well-motivated by several applications,
including the Google application described in the experimental section.

Basu et al.~\cite{Basu04} study the problem of minimizing the $k$-means objective in the presence of limited supervision. This supervision is in the form of pairwise {\em must-link} and {\em cannot-link} constraints. They propose a variation of the Lloyd's method for this problem and show promising experimental results. The split/merge requests that we study are a more natural form of interaction because they capture macroscopic properties of a cluster. Getting pairwise constraints among data points involves much more effort on the part of the user and is unrealistic in many scenarios.

The stability property that we consider is a natural generalization of the ``stable marriage'' property~(see Definition~\ref{def:strong-stability}) that has been studied in a variety of previous works~\cite{BalcanBV08, Bryant01}. It is the weakest among the stability properties that have been studied recently such as strict separation and strict threshold separation~\cite{BalcanBV08,Krishnamurthy11}.  This property is known to hold for real-world data.  In particular, ~\cite{Voevodski11}
observed that this property holds for protein sequence data, where similarities are computed with sequence alignment and ground truth clusters correspond to evolutionary-related proteins.

\section{Notation and Preliminaries}
\label{sec:notation}
Given a data set $X$ of $n$ points we define $\mathcal{C} = \lbrace C_{1},C_{2}, \ldots C_{k} \rbrace$ to be a $k$-clustering of $X$ where the $C_i$'s represent the individual clusters. Given two clusterings $\mathcal{C}$ and $\mathcal{C}'$, we define the distance between a cluster $C_i \in \mathcal{C}$ and the clustering $\mathcal{C}'$ as:
\begin{displaymath}
\dist(C_i,\mathcal{C}') = \vert \{ C'_j \in \mathcal{C}' : C'_j \cap C_i \ne \emptyset \} \vert - 1.
\end{displaymath}

This distance is the number of \emph{additional} clusters in $\mathcal{C}'$ that contain points from $C_i$; it evaluates to 0 when all points in $C_i$ are contained in a single cluster in $\mathcal{C}'$.  Naturally, we can then define the distance between $\mathcal{C}$ and $\mathcal{C}'$ as:
$
\dist(\mathcal{C},\mathcal{C}') = \sum_{C_i \in \mathcal{C}} \dist(C_i,\mathcal{C}').
$
Notice that this notion of clustering distance is asymmetric: $\dist(\mathcal{C},\mathcal{C}') \ne \dist(\mathcal{C}',\mathcal{C})$.
% For example, if $m = 1$, then it must be the case that $\dist(\mathbb{C},\mathbb{C'}) = 0$, and $\dist(\mathbb{C'},\mathbb{C}) = l - 1$.
Also note that $\dist(\mathcal{C},\mathcal{C}') = 0$ if and only if $\mathcal{C}$ refines $\mathcal{C}'$.
Observe that if $\mathcal{C}$ is the ground-truth clustering, and $\mathcal{C}'$ is a proposed clustering, then $\dist(\mathcal{C},\mathcal{C}')$ can be considered an \emph{underclustering error}, and $\dist(\mathcal{C}',\mathcal{C})$ an \emph{overclustering error}.

An underclustering error is an instance of several clusters in a proposed clustering containing points from the same ground-truth cluster; this ground-truth cluster is said to be \emph{underclustered}.  Conversely, an overclustering error is an instance of points from several ground-truth clusters contained in the same cluster in a proposed clustering; this proposed cluster is said to be \emph{overclustered}.
In the following sections we use $\mathcal{C}^{\ast} = \lbrace C^{\ast}_{1},C^{\ast}_{2}, \ldots C^{\ast}_{k} \rbrace$ to refer to the ground-truth clustering, and use $\mathcal{C}$ to refer to a proposed clustering.  We use $\delta_{u}$ to refer to the underclustering error of a proposed clustering, and $\delta_{o}$ to refer to the overclustering error.  In other words, we have $\delta_{u} = \dist(\mathcal{C}^{\ast},\mathcal{C})$ and $\delta_{o} = \dist(\mathcal{C},\mathcal{C}^{\ast})$.  We also use $\delta$ to denote the sum of the two errors: $\delta = \delta_{u} + \delta_{o}$.  We call $\delta$ the \emph{under/overclustering error}, and use the $\delta(\mathcal{C}, \mathcal{C}^{\ast})$ to refer to the error of $\mathcal{C}$ with respect to $\mathcal{C}^{\ast}$.

We also observe that we can define the distance between two clusterings using the \emph{correlation-clustering} objective function. Given a proposed clustering $\mathcal{C}$, and a ground-truth clustering $\mathcal{C}^{\ast}$, we define the correlation-clustering error $\delta_{cc}$ as the number of (ordered) pairs of points that are clustered \emph{inconsistently} with $\mathcal{C}^{\ast}$:

\begin{displaymath}
\delta_{cc} = \vert \lbrace (u,v) \in X \times X : c(u,v) \ne c^{\ast}(u,v)  \rbrace \vert,
\end{displaymath}
where $c(u,v) = 1$ if $u$ and $v$ are in the same cluster in $\mathcal{C}$, and 0 otherwise; $c^{\ast}(u,v) = 1$ if $u$ and $v$ are in the same cluster in $\mathcal{C}^{\ast}$, and 0 otherwise.

Note that we may divide the correlation-clustering error $\delta_{cc}$ into overclustering component $\delta_{cco}$ and underclustering component $\delta_{ccu}$:

\begin{displaymath}
\delta_{cco} = \vert \lbrace (u,v) \in X \times X : c(u,v) = 1 \textrm{ and } c^{\ast}(u,v) = 0  \rbrace \vert
\end{displaymath}
\begin{displaymath}
\delta_{ccu} = \vert \lbrace (u,v) \in X \times X : c(u,v) = 0 \textrm{ and } c^{\ast}(u,v) = 1  \rbrace \vert
\end{displaymath}

In our formal analysis we model the user as an oracle that provides edit requests.

\begin{definition}[Local algorithm]
 We say that an interactive clustering algorithm is {\em local} if in each iteration only the cluster assignments of points involved in the oracle request may be changed.  If the oracle requests to split $C_i$, the algorithm may only reassign the points in $C_i$.  If the oracle requests to merge $C_i$ and $C_j$, the algorithm may only reassign the points in $C_i \cup C_j$.
\end{definition}

We next formally define the properties of a clustering that we study in this work.

\begin{definition}[Stability]
\label{def:strong-stability}
Given a clustering $\mathcal{C} = \{C_1, C_2, \cdots C_k\}$ over a domain $X$ and a similarly function $S: X \times X \mapsto \Re$, we say that $\mathcal{C}$ satisfies stability with respect to $S$ if for all $i \ne j$, and for all $A \subset C_i$ and $A' \subseteq C_j$, $S(A, C_i \setminus A) > S(A,A')$, where for any two sets $A,A'$, $S(A,A') = E_{x \in A, y \in A'} S(x,y)$.
\end{definition}

 In our analysis, we assume that the ground-truth clustering satisfies stability, and we have access to the corresponding similarity function.  In addition, we also study the following stronger properties of a clustering, which were first introduced in ~\cite{BalcanBV08}.

\begin{definition}[Strict separation]
Given a clustering $C = \{C_1, C_2, \cdots C_k\}$ over a domain $X$ and a similarly function $S: X \times X \mapsto \Re$, we say that $C$ satisfies strict separation with respect to $S$ if for all $i \ne j$, $x,y \in C_i$ and $z \in C_j$, $S(x,y) > S(x,z)$.
\end{definition}

\begin{definition}[Strict threshold separation]
Given a clustering $C = \{C_1, C_2, \cdots C_k\}$ over a domain $X$ and a similarly function $S: X \times X \mapsto \Re$, we say that $C$ satisfies strict threshold separation with respect to $S$ if there exists a threshold $t$ such that, for all $i$, $x,y \in C_i$, $S(x,y) > t$, and, for all $i \ne j$, $x \in C_i, y \in C_j$, $S(x,y) \le t$.
\end{definition}

Clearly, \emph{strict separation} and \emph{strict threshold separation} imply \emph{stability}.

In order for our algorithms to make progress, the oracle requests must be somewhat consistent with the target clustering.

\begin{definition}[$\eta$-merge model]
In the $\eta$-merge model the oracle requests have the following properties

\noindent $split(C_i)$: $C_i$ contains points from two or more target clusters.

\noindent $merge(C_i , C_j)$: At least an $\eta$-fraction of the points in each $C_i$ and $C_j$ belong to the same target cluster.

\end{definition}

\begin{definition}[Unrestricted-merge model]
In the unrestricted-merge model the oracle requests have the following properties
%\begin{itemize}

\noindent $split(C_i)$: $C_i$ contains points from two or more target clusters.

\noindent $merge(C_i , C_j)$: At least $1$ point in each $C_i$ and $C_j$ belongs to the same target cluster.
%\end{itemize}
\end{definition}

Note that the assumptions about the nature of the split requests are the same in both models.  In the $\eta$-merge model, the oracle may request to merge two clusters if both have a \emph{constant fraction} of points from the same target cluster.  In the unrestricted-merge model, the oracle may request to merge two clusters if both have \emph{some} points from the same target cluster.

\subsection{Generalized clustering error}
\label{sec:generalized-clustering-error}

We observe that the clustering errors defined in the previous section may be generalized by abstracting their common properties.  We define the following properties of a \emph{natural} clustering error, which is any integer-valued error that decreases when we locally improve the proposed clustering.

\begin{definition}
\label{def:natural-clustering-error}
We say that a clustering error is \emph{natural} if it satisfies the following properties:
\begin{itemize}
\item  If there exists a cluster $C_{i}$ that contains points from $C^{\ast}_{j}$ and some other ground-truth cluster(s), then splitting this cluster into two clusters $C_{i,1} = C_{i} \cap C^{\ast}_{j}$ (which contains only points from $C^{\ast}_{j}$), and $C_{i,2} = C_{i} - C_{i,1}$ (which contains the other points) must decrease the error.
\item  If there exists two clusters that contain only points from the same target cluster, then merging them into one cluster must decrease the error.
\item The error is integer-valued.
\end{itemize}
\end{definition}

We expect a lot of definitions of clustering error to satisfy the above criteria (especially the first two properties), in addition to other domain-specific criteria. Clearly, the under/overclustering error $\delta = \delta_{u} + \delta_{o}$ and the correlation-clustering error $\delta_{cc}$ are also \emph{natural} clustering errors~(Claim~\ref{claim:natural-errors}).  As before, for a \emph{natural} clustering error $\gamma$, a proposed clustering $\mathcal{C}$ and the target clustering $\mathcal{C}^{\ast}$, we will use $\gamma(\mathcal{C}, \mathcal{C}^{\ast})$ to denote the magnitude of the error of $\mathcal{C}$ with respect to $\mathcal{C}^{\ast}$.

Moreover, it is easy to see that the under/overclustering error defined in the previous section is the lower-bound on any \emph{natural} clustering error~(Theorem~\ref{thm:generalized-clustering-error}).
\begin{claim}
\label{claim:natural-errors}
The under/overclustering error and the correlation clustering error satisfy Definition~\ref{def:natural-clustering-error} and hence are natural clustering errors.
\end{claim}

\begin{theorem}
\label{thm:generalized-clustering-error}
For any \emph{natural} clustering error $\gamma$, any proposed clustering $\mathcal{C}$, and any target clustering $\mathcal{C}^{\ast}$, $\gamma(\mathcal{C}, \mathcal{C}^{\ast}) \ge \delta(\mathcal{C}, \mathcal{C}^{\ast})$.
\end{theorem}
\begin{proof}
Given any proposed clustering $\mathcal{C}$, and any target clustering $\mathcal{C}^{\ast}$, we may transform  $\mathcal{C}$ into $\mathcal{C}^{\ast}$ via the following sequence of edits.  First, we split all over-clustering instances using the following iterative procedure: while there exists a cluster $C_{i}$ that contains points from $C^{\ast}_{j}$ and some other ground-truth cluster(s), we split it into two clusters $C_{i,1} = C_{i} \cap C^{\ast}_{j}$ and $C_{i,2} = C_{i} - C_{i,1}$.  Note that this iterative split procedure will require exactly $\delta_{o}$ split edits, where $\delta_{o}$ is the initial overclustering error.  Then, when we are left with only ``pure'' clusters (each intersects exactly one target cluster), we merge all under-clustering instances using the following iterative procedure: while there exist two clusters $C_{i}$ and $C_{j}$ that contain only points from the same target cluster, merge $C_{i}$ and $C_{j}$.  Note that this iterative merge procedure will require exactly $\delta_{u}$ merge edits, where $\delta_{u}$ is the initial underclustering error.  Let us use $\gamma$ to refer to any \emph{natural} clustering error of $\mathcal{C}$ with respect to $\mathcal{C}^{\ast}$.  By the first property of \emph{natural} clustering error, each split must have decreased $\gamma$ by at least one.  By the second property, each merge must have decreased $\gamma$ by at least one as well.  Given that we performed exactly $\delta = \delta_{o} + \delta_{u}$ edits, it follows that initially $\gamma(\mathcal{C}, \mathcal{C}^{\ast})$ must have been at least $\delta$.
\end{proof}
For additional discussion about comparing clusterings see \cite{meila07}.  Note that several criteria discussed in~\cite{meila07} satisfy our first two properties (for a similarity measure we may replace "must decrease the error" with "must increase the similarity").  In addition, the Rand and Mirkin criteria discussed in \cite{meila07} are closely related to the correlation clustering error defined here (all three measures are a function of the number of pairs of points that are clustered incorrectly).

\section{The $\eta$-merge model}
\label{sec:eta-merge}
In this section we describe and analyze the algorithms in the $\eta$-merge model.  As a pre-processing step for all our algorithms, we first run the hierarchical average-linkage algorithm on all the points in the data set to compute the global average-linkage tree, which we denote by $T_{glob}$.  The leaf nodes in this tree contain the individual points, and the root node contains all the points.  The tree is computed in a bottom-up fashion: starting with the leafs in each iteration the two most similar nodes are merged, where the similarity between two nodes $N_{1}$ and $N_{2}$ is the average similarity between points in $N_{1}$ and points in $N_{2}$.

We assign a label ``impure'' to each cluster in the initial clustering; these labels are used by the merge procedure.  Given a split or merge request, a local clustering edit is computed from the global tree $T_{glob}$ as described in Figure~\ref{fig:split-average-linkage} and Figure~\ref{fig:merge-average-linkage-relaxed}.

To implement Step 1 in Figure~\ref{fig:split-average-linkage}, we start at the root of $T_{glob}$ and ``follow'' the points in $C_i$ down one of the branches until we find a node that splits them.  In order to implement Step 2 in Figure~\ref{fig:merge-average-linkage-relaxed}, it suffices to start at the root of $T_{glob}$ and perform a post-order traversal, only considering nodes that have ``enough'' points from both clusters, and return the first output node.
\begin{figure}[h]
\caption{Split procedure}
%\vspace{-5mm}
\label{fig:split-average-linkage}
\begin{center}
\fbox{
\begin{minipage}{0.95 \columnwidth}
{\bf Algorithm}: \textsc{Split Procedure} \smallskip \\
~~{\bf Input}: Cluster $C_i$, global average-linkage tree $T_{glob}$.
\begin{enumerate}
\item Search $T_{glob}$ to find the node $N$ at which the set of points in $C_i$ are first split in two.
\item Let $N_1$ and $N_2$ be the children of $N$.  Set $C_{i,1} = N_1 \cap C_{i}$, $C_{i,2} = N_2 \cap C_{i}$.
\item Delete $C_i$ and replace it with $C_{i,1}$ and $C_{i,2}$. Mark the two new clusters as ``impure''.
\end{enumerate}
%\end{enumerate}
\end{minipage}
}
\end{center}
\end{figure}

The split procedure is fairly intuitive: if the average-linkage tree is consistent with the target clustering, it suffices to find the node in the tree where the corresponding points are first split in two.  It is more challenging to develop a correct merge procedure: note that Step 2 in Figure~\ref{fig:merge-average-linkage-relaxed} is only correct if $\eta > 0.5$, which ensures that if two nodes in the tree have more than an $\eta$-fraction of the points from $C_{i}$ and $C_{j}$, one must be an ancestor of the other.  If the average-linkage tree is consistent with the ground-truth, then clearly the node equivalent to the corresponding target cluster (that $C_{i}$ and $C_{j}$ both intersect) will have enough points from $C_{i}$ and $C_{j}$; therefore the node that we find in Step 2 must be this node or one of its descendants.  In addition, because our merge procedure replaces two clusters with three, we require pure/impure labels for the merge requests to terminate: ``pure'' clusters may only have other points added to them, and retain this label throughout the execution of the algorithm.
\begin{figure}
\caption{Merge procedure}
%\vspace{-5mm}
\label{fig:merge-average-linkage-relaxed}
\begin{center}
\fbox{
\begin{minipage}{0.95 \columnwidth}
{\bf Algorithm}: \textsc{Merge Procedure} \smallskip \\
~~{\bf Input}: Clusters $C_i$ and $C_j$, global average-linkage tree $T_{glob}$.

\begin{enumerate}
\item If $C_i$ is marked as ``pure'' set $\eta_1 = 1$ else set $\eta_1 = \eta$. Similarly set $\eta_2$ for $C_j$.
\item Search $T_{glob}$ for a node of maximal depth $N$ that contains \emph{enough} points from $C_i$ and $C_j$: $|N \cap C_i| \ge \eta_1 |C_i|$ and $|N \cap C_j| \ge \eta_2 |C_j|$.
\item Replace $C_i$ by $C_i \setminus N$, replace $C_j$ by $C_j \setminus N$.
\item Add a new cluster containing $N \cap (C_i \cup C_j)$, mark it as ``pure''.
\end{enumerate}
%\end{enumerate}
\end{minipage}
}
\end{center}
\end{figure}

We now state the performance guarantee for these split and merge algorithms.

\begin{theorem}
\label{thm:strong-stability-relaxed}
Suppose the target clustering satisfies stability, and the initial clustering has overclustering error $\delta_o$ and underclustering error $\delta_u$.  In the $\eta$-merge model, for any $\eta > 0.5$, the algorithms in Figure~\ref{fig:split-average-linkage} and Figure~\ref{fig:merge-average-linkage-relaxed} require at most $\delta_o$ split requests and $2(\delta_u + k) \log_{\frac 1 {1-\eta}} n$ merge requests to find the target clustering.
\end{theorem}

In order to prove the theorem, we must do some preliminary analysis.  First, we observe that if the target clustering satisfies stability, then every node of the average-linkage tree must be \emph{laminar} (consistent) with respect to the ground-truth clustering.

Informally, each node in a hierarchical clustering tree $T$ is \emph{laminar} (consistent) with respect to the clustering $\mathcal{C}$ if for each cluster $C_{i} \in \mathcal{C}$, the points in $C_{i}$ are first grouped together in $T$ before they are grouped with points from any other cluster $C_{j \ne i}$.  We formally state and prove these observations next.

\begin{definition}[Laminar]
A node $N$ is laminar with respect to a clustering $\mathcal{C}$ if for each cluster $C_{i} \in \mathcal{C}$ we have either $N \cap C_{i} = \emptyset$, $N \subseteq C_{i}$, or $C_{i} \subseteq N$.
\end{definition}

\begin{lemma}
\label{lem:laminar-average-linkage}
Suppose the ground-truth clustering $\mathcal{C}^{\ast}$ over a domain $X$ satisfies stability with respect to a similarity function $S$.  Let $T$ be the average-linkage tree for $X$ constructed with $S$.  Then every node in $T$ is laminar w.r.t. $\mathcal{C}^{\ast}$.
\end{lemma}

\begin{proof}
The proof of this statement can be found in~\cite{BalcanBV08}. The intuition is that if there is a node in $T$ that is not laminar w.r.t. $C^{\ast}$, then the average-linkage algorithm, at some step, must have merged $A \subset C^{\ast}_i$, with $B \subset C^{\ast}_j$ for some $i \neq j$. However, this will contradict the stability property for the sets $A$ and $B$.
\end{proof}

It follows that the split computed by the algorithm in Figure~\ref{fig:split-average-linkage} must also be consistent with the target clustering; we call such splits \emph{clean}.

\begin{definition}[Clean split]
\label{def:clean-split}
A partition (split) of a cluster $C_{i}$ into clusters $C_{i,1}$ and $C_{i,2}$ is said to be \emph{clean} if $C_{i,1}$ and $C_{i,2}$ are non-empty, and for each ground-truth cluster $C^{\ast}_{j}$ such that $C^{\ast}_{j} \cap C_{i} \ne \emptyset$, either $C^{\ast}_{j} \cap C_{i} = C^{\ast}_{j} \cap C_{i,1}$ or $C^{\ast}_{j} \cap C_{i} = C^{\ast}_{j} \cap C_{i,2}$.
\end{definition}

We now prove the correctness of the split/merge procedures.

\begin{lemma}
\label{lem:relaxed-main-lemma}
If the ground-truth clustering satisfies stability and $\eta > 0.5$ then,
\begin{itemize}
\item[\textbf{a.}] The split procedure in Figure~\ref{fig:split-average-linkage} always produces a clean split.
\item[\textbf{b.}] The new cluster added in Step 4 in Figure~\ref{fig:merge-average-linkage-relaxed} must be ``pure'', i.e., it must contain points from a single ground-truth cluster.
 \end{itemize}
\end{lemma}
\begin{proof}
\textbf{a.}  For purposes of contradiction, suppose the returned split is not clean: $C_{i,1}$ and $C_{i,2}$ contain points from the same ground-truth cluster $C^{\ast}_{j}$. It must be the case that $C_{i}$ contains points from several ground-truth clusters, which implies that w.l.o.g. $C_{i,1}$ contains points from some other ground-truth cluster ${C}^{\ast}_{l \ne j}$.  This implies that $N_1$ is not laminar w.r.t. $\mathcal{C}^{\ast}$, which contradicts Lemma~\ref{lem:laminar-average-linkage}.
\textbf{b.}  By our assumption, at least $\frac 1 2 |C_{i}|$ points from $C_{i}$ and $\frac 1 2 |C_{j}|$ points from $C_{j}$ are from the same ground-truth cluster $C^{\ast}_{l}$.  Clearly, the node $N'$ in $T_{glob}$ that is equivalent to $C^{\ast}_{l}$ (which contains all the points in $C^{\ast}_{l}$ and no other points) must contain \emph{enough} points from $C_{i}$ and $C_{j}$, and only ascendants and descendants of $N'$ may contain more than an $\eta > 1/2$ fraction of points from both clusters.  Therefore, the node $N$ that we find with a depth-first search must be $N'$ or one of its descendants, and will only contain points from $C^{\ast}_{l}$.
\end{proof}

Using the above lemma, we can prove the bounds on the split and merge requests stated in Theorem~\ref{thm:strong-stability-relaxed}. 

\begin{proof}[Proof of Theorem~\ref{thm:strong-stability-relaxed}]

We first give a bound on the number of splits.  Observe that each split reduces the overclustering error by exactly 1.  To see this, suppose we execute Split($C_{1}$), and call the resulting clusters $C_{2}$ and $C_{3}$.   Call $\delta_{1}$ the overclustering error before the split, and $\delta_{2}$ the overclustering error after the split.  Let's use $k_{1}$ to refer to the number of ground-truth clusters that intersect $C_{1}$, and define $k_{2}$ and $k_{3}$ similarly.  Due to the \emph{clean split} property, no ground-truth cluster can intersect both $C_{2}$ and $C_{3}$, therefore it must be the case that $k_{2} + k_{3} = k_{1}$.  Also, clearly $k_{2}, k_{3} > 0$.  Therefore we have:
\begin{eqnarray*}
\delta_{2} & = & \delta_{1} - (k_{1} - 1) + (k_{2} - 1) + (k_{3} - 1)\\
& = & \delta_{1} - k_{1} + (k_{2} + k_{3}) - 1\\
& = & \delta_{1} - 1.
\end{eqnarray*}

Merges cannot increase overclustering error.  Therefore the total number of splits may be at most $\delta_{o}$.  We next give the arguments about the number of impure and pure merges.

We first argue that we cannot have too many ``impure'' merges before each cluster in $C$ is marked ``pure.''  Consider the clustering $P = \{C_{i} \cap C^{\ast}_{j} \ \vert \ C_{i} \textrm{ is marked ``impure'' and } C_{i} \cap C^{\ast}_{j} \ne \emptyset\}$.  Clearly, at the start $\vert P \vert = \delta_{u} + k$.  A merge does not increase the number of clusters in $P$, and the splits do not change $P$ at all (because of the \emph{clean split} property).  Moreover, each impure merge (a merge of two impure clusters or a merge of a pure and an impure cluster) \emph{depletes} some $P_{i} \in P$ by moving $\eta \vert P_{i} \vert$ of its points to a pure cluster.  Clearly, we can then have at most $\log_{1/(1-\eta)} n$ merges depleting each $P_{i}$.  Since each impure merge must deplete some $P_{i}$, it must be the case that we can have at most  $(\delta_{u} + k) \log_{1/(1-\eta)} n$ impure merges in total.

Notice that a pure cluster can only be created by an impure merge, and there can be at most one pure cluster created by each impure merge.  Clearly, a pure merge removes exactly one pure cluster.  Therefore the number of pure merges may be at most the total number of pure clusters that are created, which is at most the total number of impure merges.  Therefore the total number of merges must be less than $2(\delta_{u} + k) \log_{1/(1-\eta)} n$.
\end{proof}

We can also restate the run-time bound in Theorem~\ref{thm:strong-stability-relaxed} in terms of any \emph{natural} clustering error $\gamma$.  The following collorary follows from Theorem~\ref{thm:strong-stability-relaxed} and Theorem~\ref{thm:generalized-clustering-error}.

\begin{corollary}
\label{corr:strong-stability-relaxed}
Suppose the target clustering satisfies stability, and the initial clustering has clustering error $\gamma$, where $\gamma$ is any \emph{natural} clustering error as defined in Definition~\ref{def:natural-clustering-error}.  In the $\eta$-merge model, for any $\eta > 0.5$, the algorithms in Figure~\ref{fig:split-average-linkage} and Figure~\ref{fig:merge-average-linkage-relaxed} require at most $O(\gamma + k) \log_{\frac 1 {1-\eta}} n$ edit requests to find the target clustering.
\end{corollary}

\subsection{Algorithms for correlation-clustering error}

To bound the number of edit requests with respect to the correlation clustering objective, we must use a different merge procedure, which is described in Figure~\ref{fig:merge-average-linkage-cc}.

\begin{figure}[htbp]
\caption{Merge procedure for the \emph{correlation-clustering} objective}
\label{fig:merge-average-linkage-cc}
\begin{center}
\fbox{
\begin{minipage}{0.95 \columnwidth}
{\bf Algorithm}: \textsc{Merge Procedure} \smallskip \\
~~{\bf Input}: Clusters $C_i$ and $C_j$, global average-linkage tree $T_{glob}$
\begin{algorithmic}
\STATE Search $T_{glob}$ for a node of maximal depth $N$ that contains \emph{enough} points from $C_i$ and $C_j$: $|N \cap C_i| \ge \eta |C_i|$ and $|N \cap C_j| \ge \eta |C_j|$
\IF{$|C_i| \ge |C_j|$}
\STATE Replace $C_i$ by $C_i \cup (N \cap C_j)$
\STATE Replace $C_j$ by $C_j \setminus N$
\ELSE
\STATE Replace $C_i$ by $C_i \setminus N$
\STATE replace $C_j$ by $C_j \cup (N \cap C_i)$
\ENDIF
\end{algorithmic}
\end{minipage}
}
\end{center}
\end{figure}

Here instead of creating a new ``pure'' cluster, we add these points to the larger of the two clusters in the merge. Notice that the new algorithm is much simpler than the merge algorithm for the under/overclustering error. Using this merge procedure and the split procedure presented earlier gives the following performance guarantee.

\begin{theorem}
\label{thm:strong-stability-cc-error}
Suppose the target clustering satisfies stability, and the initial clustering has correlation-clustering error of $\delta_{cc}$.  In the $\eta$-merge model, for any $\eta > 2/3$, using the split and merge procedures in Figures~\ref{fig:split-average-linkage} and ~\ref{fig:merge-average-linkage-cc} requires at most $\delta_{cc}$ edit requests to find the target clustering.
\end{theorem}
\begin{proof}
Consider the contributions of individual points to $\delta_{cco}$ and $\delta_{ccu}$, which are defined as:

\begin{displaymath}
\delta_{cco}(u) = \vert \lbrace v \in X : c(u,v) = 1 \textrm{ and } c^{\ast}(u,v) = 0  \rbrace \vert
\end{displaymath}
\begin{displaymath}
\delta_{ccu}(u) = \vert \lbrace v \in X : c(u,v) = 0 \textrm{ and } c^{\ast}(u,v) = 1  \rbrace \vert
\end{displaymath}

We first argue that a split of a cluster $C_{i}$ must reduce $\delta_{cc}$ by at least 1.  Given that the split is \emph{clean}, it is easy to verify that the outcome may not increase $\delta_{ccu}(u)$ for any $u \in C_{i}$.  We can also verify that for each $u \in C_{i}$, $\delta_{cco}(u)$ must decrease by at least 1.  This completes the argument, given that the correlation-clustering error with respect to all other pairs of points must remain the same.

We now argue that if $\eta > 2/3$, each merge of $C_{i}$ and $C_{j}$ must reduce $\delta_{cc}$ by at least 1.  Without loss of generality, suppose that $\vert C_{i} \vert \ge \vert C_{j} \vert$, and let us use $P$ to refer to the ``pure'' subset of $C_{j}$ that is moved to $C_{i}$.  We observe that the outcome must remove at least $\delta_{1}$ pairwise correlation-clustering errors, where $\delta_{1}$ satisfies $\delta_{1} \ge 2 \vert P \vert (\eta \vert C_{i} \vert)$.  Similarly, we observe that the outcome may add at most $\delta_{2}$ pairwise correlation-clustering errors, where $\delta_{2}$ satisfies:
\begin{displaymath}
\delta_{2} \le 2 \vert P \vert ((1-\eta) \vert C_{i} \vert) + 2 \vert P \vert ((1-\eta) \vert C_{j} \vert) \le 4 \vert P \vert ((1-\eta) \vert C_{i} \vert).
\end{displaymath}
It follows that for $\eta > 2/3$, $\delta_{1}$ must exceed $\delta_{2}$; therefore the sum of the pairwise correlation-clustering errors must decrease, giving a lower correlation-clustering error total.
\end{proof}

Observe that the runtime bound in Theorem~\ref{thm:strong-stability-cc-error} is tight: in some instances any \emph{local} algorithm requires at least $\delta_{cc}$ edits to find the target clustering.  To verify this, suppose the target clustering is composed of $n$ singleton clusters, and the initial clustering contains $n/2$ clusters of size 2.  In this instance, the initial correlation clustering error $\delta_{cc} = n/2$, and the oracle must issue at least $n/2$ split requests before we reach the target clustering (no matter how the algorithm reassigns the corresponding points).

\subsection{Algorithms under stronger assumptions}

When the data satisfies stronger stability properties we may simplify the presented algorithms and/or obtain better performance guarantees.  In particular, if the data satisfies the \emph{strict separation} property from ~\cite{BalcanBV08}, we may change the split and merge algorithms to use the local average-linkage tree, which is constructed from only the points in the edit request.  In addition, if the data satisfies \emph{strict threshold separation}, we may remove the restriction on $\eta$ and use a different merge procedure that is correct for any $\eta > 0$.

\begin{theorem}
\label{thm:strict-separation}
Suppose the target clustering satisfies strict separation, and the initial clustering has overclustering error $\delta_o$ and underclustering error $\delta_u$.  In the $\eta$-merge model, for any $\eta > 0.5$,  the algorithms in Figure~\ref{fig:local-split-average-linkage} and Figure~\ref{fig:local-merge-average-linkage-relaxed} require at most $\delta_o$ split requests and $2(\delta_u + k) \log_{\frac 1 {1-\eta}} n$ merge requests to find the target clustering.
\end{theorem}

\begin{proof}
Let us use $\mathcal{L}^{\ast}$ to refer to the ground-truth clustering of the points in the split/merge request.  If the target clustering satisfies strict separation, it is easy to verify that every node in the local average-linkage tree $T_{loc}$ must be laminar (consistent) w.r.t. $\mathcal{L}^{\ast}$.  We can then use this observation to prove the equivalent of Lemma~\ref{lem:relaxed-main-lemma} for the split procedure in Figure~\ref{fig:local-split-average-linkage} and the merge procedure in Figure~\ref{fig:local-merge-average-linkage-relaxed}.  The analysis in Theorem~\ref{thm:strong-stability-relaxed} remains unchanged.
\end{proof}

\begin{theorem}
\label{thm:strict-threshold-separation}
Suppose the target clustering satisfies strict threshold separation, and the initial clustering has overclustering error $\delta_o$ and underclustering error $\delta_u$.  In the $\eta$-merge model, for any $\eta > 0$,  the algorithms in Figure~\ref{fig:local-split-average-linkage} and Figure~\ref{fig:merge-strict-threshold-separation} require at most $\delta_o$ split requests and $2(\delta_u + k) \log_{\frac 1 {1-\eta}} n$ merge requests to find the target clustering.
\end{theorem}

\begin{proof}
If the target clustering satisfies strict threshold separation, we can verify that the split procedure in Figure~\ref{fig:local-split-average-linkage} and the merge procedure in Figure~\ref{fig:merge-strict-threshold-separation} are correct for any $\eta >0$.  The analysis in Theorem~\ref{thm:strong-stability-relaxed} remains unchanged.

To verify that the split procedure always produces a clean split, again let us use $\mathcal{L}^{\ast}$ to refer to the ground-truth clustering of the points in the split request.  We can again verify that each node in the local average-linkage tree $T_{loc}$ must be laminar (consistent) w.r.t. $\mathcal{L}^{\ast}$.  It follows that the split procedure always produces a clean split.  Note that clearly this argument does not depend on the setting of $\eta$.

We now verify that the new cluster added by the merge procedure Figure~\ref{fig:merge-strict-threshold-separation} must be ``pure'' (must contain points from a single target cluster).  To see this, observe that in the graph $G$ in Figure~\ref{fig:merge-strict-threshold-separation}, all pairs of points from the same target cluster are connected before any pairs of points from different target clusters.  It follows that the first component that contains at least an $\eta$-fraction of points from $C_{i}$ and $C_{j}$ must be ``pure''.  Note that this argument applies for any $\eta > 0$.
\end{proof}

Note that the merge procedure in Figure~\ref{fig:merge-strict-threshold-separation} is correct for $\eta \le 0.5$ only if the target clustering satisfies \emph{strict threshold separation}: there is a single threshold $t$ such that for all $i$, $x,y \in C^{\ast}_i$, $S(x,y) > t$, and, for all $i \ne j$, $x \in C^{\ast}_i, y \in C^{\ast}_j$, $S(x,y) \le t$.  When only \emph{strict separation} holds (the threshold for each target cluster may be different), this procedure may first connect points from different target clusters, and for $\eta \le 0.5$ this component may then be large enough to be output.

As in Corollary~\ref{corr:strong-stability-relaxed}, we may also restate the run-time bounds in Theorem~\ref{thm:strict-separation} and Theorem~\ref{thm:strict-threshold-separation} in terms of any natural clustering error $\gamma$. The following corollaries follow from Theorem~\ref{thm:strict-separation}, Theorem~\ref{thm:strict-threshold-separation} and Theorem~\ref{thm:generalized-clustering-error}.

\begin{corollary}
\label{corr:strict-separation}
Suppose the target clustering satisfies strict separation, and the initial clustering has clustering error $\gamma$, where $\gamma$ is any \emph{natural} clustering error as defined in Definition~\ref{def:natural-clustering-error}.  In the $\eta$-merge model, for any $\eta > 0.5$,  the algorithms in Figure~\ref{fig:local-split-average-linkage} and Figure~\ref{fig:local-merge-average-linkage-relaxed} require at most $O(\gamma + k) \log_{\frac 1 {1-\eta}} n$ edit requests to find the target clustering.
\end{corollary}

\begin{corollary}
\label{corr:strict-threshold-separation}
Suppose the target clustering satisfies strict threshold separation, and the initial clustering has clustering error $\gamma$, where $\gamma$ is any \emph{natural} clustering error as defined in Definition~\ref{def:natural-clustering-error}.  In the $\eta$-merge model, for any $\eta > 0$,  the algorithms in Figure~\ref{fig:local-split-average-linkage} and Figure~\ref{fig:merge-strict-threshold-separation} require at most $O(\gamma + k) \log_{\frac 1 {1-\eta}} n$ edit requests to find the target clustering.
\end{corollary}

\begin{figure}[htbp]
\caption{Split procedure under stronger assumptions}
%\vspace{-5mm}
\label{fig:local-split-average-linkage}
\begin{center}
\fbox{
\begin{minipage}{0.95 \columnwidth}
{\bf Algorithm}: \textsc{Split Procedure} \smallskip \\
~~{\bf Input}: Cluster $C_i$, local average-linkage tree $T_{loc}$.
\begin{enumerate}
\item Let $C_{i,1}$ and $C_{i,2}$ be the children of the root in $T_{loc}$.
\item Delete $C_i$ and replace it with $C_{i,1}$ and $C_{i,2}$. Mark the two new clusters as ``impure''.
\end{enumerate}
%\end{enumerate}
\end{minipage}
}
\end{center}
\end{figure}

\begin{figure}[htbp]
\caption{Merge procedure under strict separation}
%\vspace{-5mm}
\label{fig:local-merge-average-linkage-relaxed}
\begin{center}
\fbox{
\begin{minipage}{0.95 \columnwidth}
{\bf Algorithm}: \textsc{Merge Procedure} \smallskip \\
~~{\bf Input}: Clusters $C_i$ and $C_j$, local average-linkage tree $T_{loc}$.

\begin{enumerate}
\item If $C_i$ is marked as ``pure'' set $\eta_1 = 1$ else set $\eta_1 = \eta$. Similarly set $\eta_2$ for $C_j$.
\item Search $T_{loc}$ for a node of maximal depth $N$ that contains \emph{enough} points from $C_i$ and $C_j$: $|N \cap C_i| \ge \eta_1 |C_i|$ and $|N \cap C_j| \ge \eta_2 |C_j|$.
\item Replace $C_i$ by $C_i \setminus N$, replace $C_j$ by $C_j \setminus N$.
\item Add a new cluster containing $N \cap (C_i \cup C_j)$, mark it as ``pure''.
\end{enumerate}
%\end{enumerate}
\end{minipage}
}
\end{center}
\end{figure}

\begin{figure}
\caption{Merge procedure under strict threshold separation}
\label{fig:merge-strict-threshold-separation}
\begin{center}
\fbox{
\begin{minipage}{0.95 \columnwidth}
{\bf Algorithm}: \textsc{Merge Procedure} \smallskip \\
~~{\bf Input}: Clusters $C_i$ and $C_j$.

\begin{enumerate}
\item If $C_i$ is marked as ``pure'' set $\eta_1 = 1$ else set $\eta_1 = \eta$. Similarly set $\eta_2$ for $C_j$.
\item Let $G = (V,E)$ be a graph where $V = C_{i} \cup C_{j}$ and $E = \emptyset$.  Set $N = \emptyset$.
\item While true:
\begin{algorithmic}
\STATE Connect the next-closest pair of points in G;
\STATE Let $\hat{C_{1}}, \hat{C_{2}}, \ldots , \hat{C_{m}}$ be the connected components of $G$;
\IF{there exists $\hat{C_{l}}$ such that $\vert \hat{C_{l}} \cap C_{i} \vert \ge \eta \vert C_{i} \vert$ and $\vert \hat{C_{l}} \cap C_{j} \vert \ge \eta \vert C_{j} \vert$}
\STATE $N = \hat{C_{l}}$;
\STATE break;
\ENDIF
\end{algorithmic}
\item Replace $C_i$ by $C_i \setminus N$, replace $C_j$ by $C_j \setminus N$.
\item Add a new cluster containing $N$, mark it as ``pure''.
\end{enumerate}
%\end{enumerate}
\end{minipage}
}

\end{center}
\end{figure}

\eat{
We would however like to stress that even under these stronger properties it is non-trivial to design a local algorithm in our model. First of all we do not assume the knowledge of $k$ apriori. Secondly, previous work~\cite{BBV08} would only imply that under such stronger conditions it is possible to construct a hierarchical tree containing clustering. Hence if we can ignore the original clustering altogether, then with (k-1) split requests we can compute the target clustering by starting from the root node of this hierarchical tree. However, here we consider the more realistic scenario where we do not have control over the original clustering, and already start with a fairly good clustering that may only be edited locally.}

\section{The unrestricted-merge model}
\label{sec:unrestricted-merge}
In this section we further relax the assumptions about the nature of the oracle requests.  As before, the oracle may request to split a cluster if it contains points from two or more target clusters.  For merges, now the oracle may request to merge $C_i$ and $C_j$ if both clusters contain only a single point from the same ground-truth cluster.  We note that this is a minimal set of assumptions for a local algorithm to make progress, otherwise the oracle may always propose irrelevant splits or merges that cannot reduce clustering error.  For this model we propose the merge algorithm described in Figure~\ref{fig:merge-average-linkage-unrestricted}. The split algorithm remains the same as in Figure~\ref{fig:split-average-linkage}.

\begin{figure}[htbp]
\caption{Merge procedure for the unrestricted-merge model}
%\vspace{-1mm}
\label{fig:merge-average-linkage-unrestricted}
\begin{center}
\fbox{
\begin{minipage}{0.95 \columnwidth}
{\bf Algorithm}: \textsc{Merge Procedure} \smallskip \\
~~{\bf Input}: Clusters $C_i$ and $C_j$, global average-linkage tree $T_{avg}.$

\begin{enumerate}
\item Let $C'_i$, $C'_j$ = Split($C_i \cup C_j$), where the split is performed as in the previous section.
\item Delete $C_i$ and $C_j$.
\item If the sets $C'_i$ and $C'_j$ are the same as $C_i$ and $C_j$, then add $C_i \cup C_j$, otherwise add $C'_i$ and $C'_j$.
\end{enumerate}
%\end{enumerate}
\end{minipage}
}

\end{center}
\end{figure}

To provably find the ground-truth clustering in this setting we require that each merge request must be chosen uniformly at random from the set of feasible merges.  This assumption is consistent with the observation in~\cite{AwasthiZ10} that in the unrestricted-merge model with arbitrary request sequences, even very simple cases~(ex. union of intervals on a line) require a prohibitively large number of requests.  We do not make additional assumptions about the nature of the split requests; in each iteration any feasible split may be proposed by the oracle.
In this setting our algorithms have the following performance guarantee.

\begin{theorem}
\label{thm:strict-separation-random}
Suppose the target clustering satisfies stability, and the initial clustering has overclustering error $\delta_o$ and underclustering error $\delta_u$.  In the unrestricted-merge model, with probability at least $1-\epsilon$, the algorithms in Figure~\ref{fig:split-average-linkage} and Figure~\ref{fig:merge-average-linkage-unrestricted} require $\delta_o$ split requests and $O(\log \frac k {\epsilon} {{\delta}^2_u})$ merge requests to find the target clustering.
\end{theorem}

The above theorem is proved in a series of lemmas.  We first state a lemma regarding the correctness of the Algorithm in Figure~\ref{fig:merge-average-linkage-unrestricted}.  We argue that if the algorithm merges $C_{i}$ and $C_{j}$, it must be the case that both $C_{i}$ and $C_{j}$ only contain points from the same ground-truth cluster.

\begin{lemma}
\label{lem:unrestricted-pure-merge}
If the algorithm in Figure~\ref{fig:merge-average-linkage-unrestricted} merges $C_{i}$ and $C_{j}$ in Step 3, it must be the case that $C_{i} \subset C^{\ast}_{l}$ and $C_{j} \subset C^{\ast}_{l}$ for some ground-truth cluster $C^{\ast}_{l}$.
\end{lemma}
\begin{proof}
We prove the contrapositive.  Suppose $C_{i}$ and $C_{j}$ both contain points from $C^{\ast}_{l}$, and in addition $C_{i} \cup C_{j}$ contains points from some other ground-truth cluster.  Let us define $S_{1} = C^{\ast}_{l} \cap C_{i}$ and $S_{2} = C^{\ast}_{l} \cap C_{j}$.  Because the clusters $C'_{i}$, $C'_{j}$ result from a \emph{clean} split, it follows that $S_{1}, S_{2} \subseteq C'_{i}$ or $S_{1}, S_{2} \subseteq C'_{j}$.  Without loss of generality, assume $S_{1}, S_{2} \subseteq C'_{i}$.  Then clearly $C'_{i} \ne C_{i}$ and $C'_{i} \ne C_{j}$, so $C_{i}$ and $C_{j}$ are not merged.
\end{proof}

The $\delta_{o}$ bound on the number of split requests follows from the observation that each split reduces the overclustering error by exactly 1 (as before), and the fact that the merge procedure does not increase overclustering error.

\begin{lemma}
\label{lem:unrestricted-overclustering-error}
The merge algorithm in Figure~\ref{fig:merge-average-linkage-unrestricted} does not increase overclustering error.
\end{lemma}
\begin{proof}
Suppose $C_{i}$ and $C_{j}$ are not both ``pure'' (one or both contain elements from several ground-truth clusters), and hence we obtain two new clusters $C'_{i}$, $C'_{j}$.  Let us call $\delta_{1}$ the overclustering error before the merge, and $\delta_{2}$ the overclustering error after the merge.  Let's use $k_{1}$ to refer to the number of ground-truth clusters that intersect $C_{i}$, $k_{2}$ to refer to the number of ground-truth clusters that intersect $C_{j}$, and define $k_{1}'$ and $k_{2}'$ similarly.  The new clusters $C'_{i}$ and $C'_{j}$ result from a ``clean'' split, therefore no ground-truth cluster may intersect both of them.  It follows that $k_{1}' + k_{2}' \le k_{1} + k_{2}$.  Therefore we now have:
\begin{eqnarray*}
\delta_{2} & = & \delta_{1} - (k_{1} - 1) - (k_{2} - 1) + (k_{1}' - 1) + (k_{2}' - 1)\\
& = & \delta_{1} - (k_{1} + k_{2}) + (k_{1}' + k_{2}') \le \delta_{1}.
\end{eqnarray*}
If $C_{i}$ and $C_{j}$ are both ``pure'' (both are subsets of the same ground-truth cluster), then clearly the merge operation has no effect on the overclustering error.
\end{proof}

The following lemmas bound the number of impure and pure merges.  Here we call a proposed merge \emph{pure} if both clusters are subsets of the same ground-truth cluster, and \emph{impure} otherwise.

\begin{lemma}
\label{lem:merge-unrestricted}
The merge algorithm in Figure~\ref{fig:merge-average-linkage-unrestricted} requires at most $\delta_{u}$ impure merge requests.
\end{lemma}
\begin{proof}
We argue that the result of each impure merge request must reduce the underclustering error by at least 1.  Suppose the oracle requests to merge $C_{i}$ and $C_{j}$, and $C'_{i}$ and $C'_{j}$ are the resulting clusters.  Clearly, the local edit has no effect on the underclustering error with respect to target clusters that do not intersect $C_{i}$ or $C_{j}$.  In addition, because the new clusters $C'_{i}$ and $C'_{j}$ result from a \emph{clean} split, for target clusters that intersect exactly one of $C_{i}$, $C_{j}$, the underclustering error must stay the same.  For target clusters that intersect both $C_{i}$ and $C_{j}$, the underclustering error must decrease by exactly one; the number of such target clusters is at least one.
\end{proof}

\begin{lemma}
\label{lem:merge-prob-unrestricted}
The probability that the algorithm in Figure~\ref{fig:merge-average-linkage-unrestricted} requires more than $O(\log \frac{k}{\epsilon} {\delta}^2_{u})$ pure merge requests is less than $\epsilon$.
\end{lemma}
\begin{proof}
We first consider the pure merge requests involving points from some ground-truth cluster $C^{\ast}_{i}$, the total number of pure merge requests (involving any ground-truth cluster) can then be bounded with a union-bound.

To facilitate our argument, let us assign an identifier to each cluster containing points from $C^{\ast}_{i}$ in the following manner:
\begin{enumerate}
\item Maintain a CLUSTER-ID variable, which is initialized to 1.
\item To assign a ``new'' identifier to a cluster, set its identifier to CLUSTER-ID, and increment CLUSTER-ID.
\item In the initial clustering, assign a \emph{new} identifier to each cluster containing points from $C^{\ast}_{i}$.
\item When we split a cluster containing points from $C^{\ast}_{i}$, assign its identifier to the newly-formed cluster containing points from $C^{\ast}_{i}$.
\item When we merge two clusters and one or both of them are impure, if one of the clusters contains points from $C^{\ast}_{i}$, assign its identifier to the newly-formed cluster containing points from $C^{\ast}_{i}$.  If both clusters contain points from $C^{\ast}_{i}$, assign a \emph{new} identifier to the newly-formed cluster containing points from $C^{\ast}_{i}$.
\item When we merge two clusters $C_{1}$ and $C_{2}$, and both contain only points from $C^{\ast}_{i}$, if the outcome is one new cluster, assign it a \emph{new} identifier.  If the outcome is two new clusters, assign them the identifiers of $C_{1}$ and $C_{2}$.
\end{enumerate}

Clearly, when clusters containing points from $C^{\ast}_{i}$ are assigned identifiers in this manner, the maximum value of CLUSTER-ID is bounded by $O(\delta_{i})$, where $\delta_{i}$ denotes the underclustering error of the initial clustering with respect to $C^{\ast}_{i}$: $\delta_{i} = \dist(C^{\ast}_{i}, C)$.  To verify this, consider that we assign exactly $\delta_{i}+1$ new identifiers in Step-3, and each time we assign a new identifier in Steps 5 and 6, the underclustering error of the edited clustering with respect to $C^{\ast}_{i}$ decreases by one.

We say that a \emph{pure} merge request involving points from $C^{\ast}_{i}$ is \emph{original} if the user has never asked us to merge clusters with the given identifiers, otherwise we say that this merge request is \emph{repeated}.  Given that the maximum value of CLUSTER-ID is bounded by $O(\delta_{i})$, the total number of \emph{original} merge requests must be $O(\delta_{i}^{2})$.  We now argue that if a merge request is not original, we can lower bound the probability that it will result in the merging of the two clusters.

For repeated merge request $M_{i} = Merge(C_{1},C_{2})$, let $X_{i}$ be a random variable defined as follows:
\begin{displaymath}
X_{i} = \left\{ \begin{array}{ll}
1 & \textrm{if neither $C_{1}$ nor $C_{2}$ have been involved in}\\
& \textrm{a merge request since the last time a merge of}\\
& \textrm{clusters with these identifiers was proposed.}\\
0 & \textrm{otherwise.}\\
\end{array} \right.
\end{displaymath}

Clearly, when $X_{i} = 1$ it must be the case that $C_{1}$ and $C_{2}$ are merged.  We observe that $\textrm{Pr} \lbrack X_{i} = 1 \rbrack > \frac{1}{2 \delta_{i} + 1}$.  To verify this, observe that in each step the probability that the user requests to merge $C_{1}$ and $C_{2}$ is $\frac{1}{m}$, and the probability that the user requests to merge $C_{1}$ or $C_{2}$ with some other cluster is less than $\frac{2 \delta_{i}}{m}$, where $m$ is the total number of possible merge requests; we can then bound the probability that the former happens before the latter.

We can then use a Chernoff bound to argue that after $t = O(\log \frac{k}{\epsilon} \delta_{i}^{2})$ \emph{repeated} merge requests, the probability that $\sum_{i=1}^{t}X_{i} < \delta_{i}$ (which must be true if we need more \emph{repeated} merge requests) is less than $\epsilon/k$.  Therefore, the probability that we need more than $O(\log \frac{k}{\epsilon} \delta_{i}^{2})$ \emph{repeated} merge requests is less than $\epsilon/k$.

By the union-bound, the probability that we need more than $O(\log \frac{k}{\epsilon} \delta_{i}^{2})$ \emph{repeated} merge requests for \emph{any} ground-truth cluster $C^{\ast}_{i}$ is less than $k \cdot \epsilon/k = \epsilon$.  Therefore with probability at least $1 - \epsilon$ for all ground-truth clusters we need $\sum_{i} O(\log \frac{k}{\epsilon} \delta_{i}^{2}) = O (\log \frac{k}{\epsilon} \sum_{i} \delta_{i}^{2}) = O(\log \frac{k}{\epsilon} \delta_{u}^2)$ \emph{repeated} merge requests, where $\delta_{u}$ is the underclustering error of the original clustering.  Similarly, for all ground-truth clusters we need $\sum_{i} O(\delta_{i}^{2}) = O(\delta_{u}^{2})$ \emph{original} merge requests.  Adding the two terms together, it follows that with probability at least $1-\epsilon$ we need a total of $O(\log \frac{k}{\epsilon} \delta_{u}^{2})$ pure merge requests.
\end{proof}

As in the previous section, we also restate the run-time bound in Theorem~\ref{thm:strict-separation-random} in terms of any \emph{natural} clustering error $\gamma$.  The following collorary follows from Theorem~\ref{thm:strict-separation-random} and Theorem~\ref{thm:generalized-clustering-error}.

\begin{corollary}
\label{thm:strict-separation-random-generalized}
Suppose the target clustering satisfies stability, and the initial clustering has clustering error $\gamma$, where $\gamma$ is any \emph{natural} clustering error as defined in Definition~\ref{def:natural-clustering-error}.  In the unrestricted-merge model, with probability at least $1-\epsilon$, the algorithms in Figure~\ref{fig:split-average-linkage} and Figure~\ref{fig:merge-average-linkage-unrestricted} require $O(\log \frac k {\epsilon} {{\gamma}^2})$ edit requests to find the target clustering.
\end{corollary}

As in the previous section, if the data satisfies \emph{strcit separation}, then instead of the split procedure in Figure~\ref{fig:split-average-linkage} we can use the procedure in Figure~\ref{fig:local-split-average-linkage}, which uses the local average-linkage tree (constructed from only the points in the user request).  We can then obtain the same performance guarantee as in Theorem~\ref{thm:strict-separation-random} for the algorithms in Figure~\ref{fig:local-split-average-linkage} and Figure~\ref{fig:merge-average-linkage-unrestricted}.

\section{Experimental Results}
\label{sec:experiments}

We perform two sets of experiments: we first test the proposed split procedure on the clustering of business listings maintained by Google, and also test the proposed framework in its entirety on the much smaller newsgroup documents data set.

\subsection{Clustering business listings}

Google maintains a large collection of data records representing businesses.  These records are clustered using a similarity function; each cluster should contain records about the same distinct business; each cluster is summarized and served to users online via various front-end applications.  Users report bugs such as ``you are displaying the name of one business, but the address of another" (caused by over-clustering), or ``a particular business is shown multiple times" (caused by under-clustering).  These bugs are routed to operators who examine the contents of the corresponding clusters, and request splits/merges accordingly.  The clusters involved in these requests may be quite large and usually contain records about several businesses.  Therefore automated tools that can perform the requested edits are very helpful.

In particular, here we evaluate the effectiveness of our proposed split procedure in computing correct cluster splits.  We consider a binary split correct if the two resulting sub-clusters are ``clean'' using Definition~\ref{def:clean-split}, and consider the split incorrect otherwise.  Note that a clean split is sufficient and necessary for reducing the under/overclustering error.  To compute the splits, we use the algorithm in Figure~\ref{fig:local-split-average-linkage}, which we refer to as \emph{Clean-Split}.  This algorithm is easier to implement and run than the algorithm in Figure~\ref{fig:split-average-linkage} because we do not need to compute the global average-linkage tree.  But it is still provably correct under stronger assumptions on the data~(see Theorem~\ref{thm:strict-separation} and Theorem~\ref{thm:strict-threshold-separation}).

For comparison purposes, we use two well-known techniques for computing binary splits: the optimal 2-median clustering (\emph{2-Median}), and a ``sweep'' of the second-smallest eigenvector of the corresponding Laplacian matrix.  Let $\{v_{1},\ldots,v_{n}\}$ be the order of the vertices when sorted by their eigenvector entries, we compute the partition $\{v_{1},\ldots,v_{i}\}$ and $\{v_{i+1},\ldots,v_{n}\}$ such that its conductance is smallest (\emph{Spectral-Balanced}), and a partition such that the similarity between $v_{i}$ and $v_{i+1}$ is smallest (\emph{Spectral-Gap}).

\begin{table}[ht]
\caption{Number of correct (clean) splits}
\centering
\begin{tabular}{c c c c} % centered columns (4 columns)
\hline\hline %inserts double horizontal lines
Clean-Split & 2-Median & Spectral-Gap & Spectral-Balanced \\ [0.5ex] % inserts table
%heading
\hline % inserts single horizontal line
19 & 13 & 12 & 3 \\ [0.5ex] % [1ex] adds vertical space
\hline %inserts single line
\end{tabular}
\label{table:split-known-over-clusters} % is used to refer this table in the text
\end{table}

We compare the split procedures on 20 over-clusters that were discovered during a clustering-quality evaluation\footnote{the data set is available at \href{http://voevodski.org/data/businessListingsDatasets/description.html}{voevodski.org/data/businessListingsDatasets/description.html}.}. The results are presented in Table~\ref{table:split-known-over-clusters}.  We observe that the \emph{Clean-Split} algorithm works best, giving a correct split in 19 out of the 20 cases.  The well-known \emph{Spectral-Balanced} technique usually does not give correct splits for this application.  The balance constraint usually causes it to put records about the same business on both sides of the partition (especially when all the ``clean'' splits are not well-balanced), which increases clustering error.  As expected, the \emph{Spectral-Gap} technique improves on this limitation (because it does not have a balance constraint), but the result often still increases clustering error.  The \emph{2-Median} algorithm performs fairly well, but it may not be the right technique for this problem: the optimal centers may correspond to listings about the same business, and even if they represent distinct businesses, the resulting partition is still sometimes incorrect.

\begin{table}[ht]
\caption{Change in correlation-clustering error}
\centering
\begin{tabular}{c c c} % centered columns (4 columns)
\hline\hline %inserts double horizontal lines
Dataset & Clean-Split & 2-Median \\ [0.5ex] % inserts table
%heading
\hline % inserts single horizontal line
1 & -14 & -14 \\ [0.5ex] % [1ex] adds vertical space
2 & -5 & -5 \\ [0.5ex] % [1ex] adds vertical space
3 & -11 & -11 \\ [0.5ex] % [1ex] adds vertical space
4 & -117 & -117 \\ [0.5ex] % [1ex] adds vertical space
5 & -42 & +90 \\ [0.5ex] % [1ex] adds vertical space
6 & -4 & -4 \\ [0.5ex] % [1ex] adds vertical space
7 & -12 & -30 \\ [0.5ex] % [1ex] adds vertical space
8 & -27 & -27 \\ [0.5ex] % [1ex] adds vertical space
9 & -6 & -6 \\ [0.5ex] % [1ex] adds vertical space
10 & -6 & -6 \\ [0.5ex] % [1ex] adds vertical space
11 & +6 & -8 \\ [0.5ex] % [1ex] adds vertical space
12 & -10 & +14 \\ [0.5ex] % [1ex] adds vertical space
13 & -6 & -6 \\ [0.5ex] % [1ex] adds vertical space
14 & -12 & -22 \\ [0.5ex] % [1ex] adds vertical space
15 & -6 & -6 \\ [0.5ex] % [1ex] adds vertical space
16 & -10 & +14 \\ [0.5ex] % [1ex] adds vertical space
17 & -11 & -27 \\ [0.5ex] % [1ex] adds vertical space
18 & -10 & -10 \\ [0.5ex] % [1ex] adds vertical space
19 & -11 & -5 \\ [0.5ex] % [1ex] adds vertical space
20 & -10 & -10 \\ [0.5ex] % [1ex] adds vertical space
\hline %inserts single line
\end{tabular}
\label{table:split-known-over-clusters-cc-error} % is used to refer this table in the text
\end{table}

In addition to using the clean-split criterion, we also evaluate the computed splits using the correlation-clustering (cc) error.  We find that using this criterion \emph{Clean-Split} and \emph{2-Median} compute the best splits, while the other two algorithms perform significantly worse.  The results for \emph{Clean-Split} and \emph{2-Median} are presented in Table~\ref{table:split-known-over-clusters-cc-error}.  Note that a clean split is sufficient to reduce the correlation-clustering error, but it is not necessary.  Our experiments illustrate these observations: \emph{Clean-Split} makes progress in reducing the cc-error in 19 out of 20 cases (when the resulting split is clean), while \emph{2-Median} is able to  still reduce the cc-error even when the resulting split is not clean.  Overall, in 12 instances the two algorithms give a tie in performance; in 4 instances \emph{Clean-Split} makes more progress in reducing the correlation-clustering error; and in 4 instances \emph{2-Median} makes more progress.  Also note that \emph{Clean-Split} fails to reduce the cc-error only once; while \emph{2-Median} fails to reduce the cc-error 4 times.

\subsection{Clustering newsgroup documents}
In order to test our entire framework (the iterative application of our
algorithms), we perform computational experiments on newsgroup documents
data.\footnote{http://people.csail.mit.edu/jrennie/20Newsgroups/} The objects in these data sets are posts to twenty different online
forums (newsgroups).  We sample these data to compute 5 data sets of manageable
size (containing 276-301 elements), which are labeled A through E in the figures.  Each data set contains some documents from every newsgroup.  

Each post/document is represented by a term frequency - inverse
document frequency (tf-idf) vector~\cite{SB88}.  We use cosine similarity to
compare these vectors, which
gives a similarity measure between 0 and 1 (inclusive).  We compute an initial clustering by using the following procedure to perturb the ground-truth:
for each document we keep its ground-truth cluster assignment with probability $0.5$, and
otherwise reassign it to one of the other clusters, which is chosen uniformly at random.

In each iteration, we compute the set of all feasible splits and merges: a split of
a cluster is feasible if it contains points from 2 or more ground-truth clusters, and a merge is feasible if at least an $\eta$-
fraction of points in each cluster are from the same ground-truth cluster.
Then, we choose one of the
feasible edits uniformly at random, and ask the algorithm to compute the
corresponding edit. We continue this process until we find the ground-truth clustering or we reach 20000 iterations.  Note that for the $\eta$-merge model, our theoretical analysis is applicable to any edit-request sequence, but in our experiments for simplicity we still
select a feasible edit uniformly at random.

Our initial clusterings have over-clustering error of about 100, under-clustering error of about 100; and correlation-clustering error of about 5000.

We notice that for newsgroup documents it is difficult to compute average-linkage
trees that are very consistent with the ground-truth.  This observation was also made in other
clustering studies that report that the hierarchical trees constructed from these data have low
purity ~\cite{Telgarsky12, HellerG05}.  These observations suggest that these data are quite challenging for clustering algorithms.  To test how
well our algorithms can perform with better data, we prune the data sets
by repeatedly finding
the outlier in each target cluster and removing it, where the outlier is
the point with minimum sum-similarity to the other points in the target
cluster.
For each data set, we perform experiments with the original (unpruned) data
set, a pruned data set with 2 points removed per target cluster, and a
pruned
data set with 4 points removed per target cluster, which prunes 40 and 80
points, respectively (given
that we have 20 target clusters).

\begin{figure}
\centering
  \subfigure{\includegraphics[width=60mm,height=40mm]{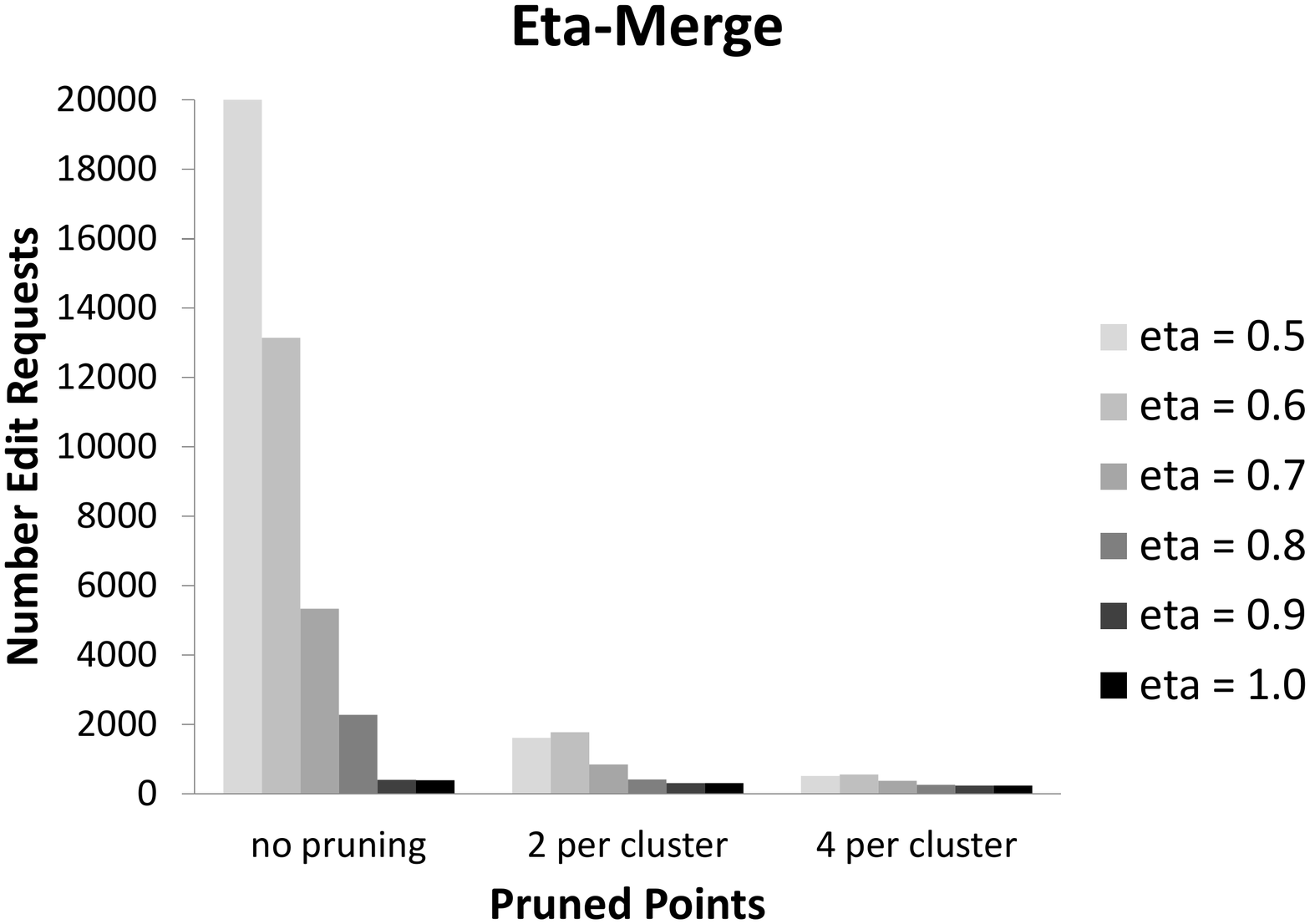}}
  \subfigure{\includegraphics[width=60mm,height=40mm]{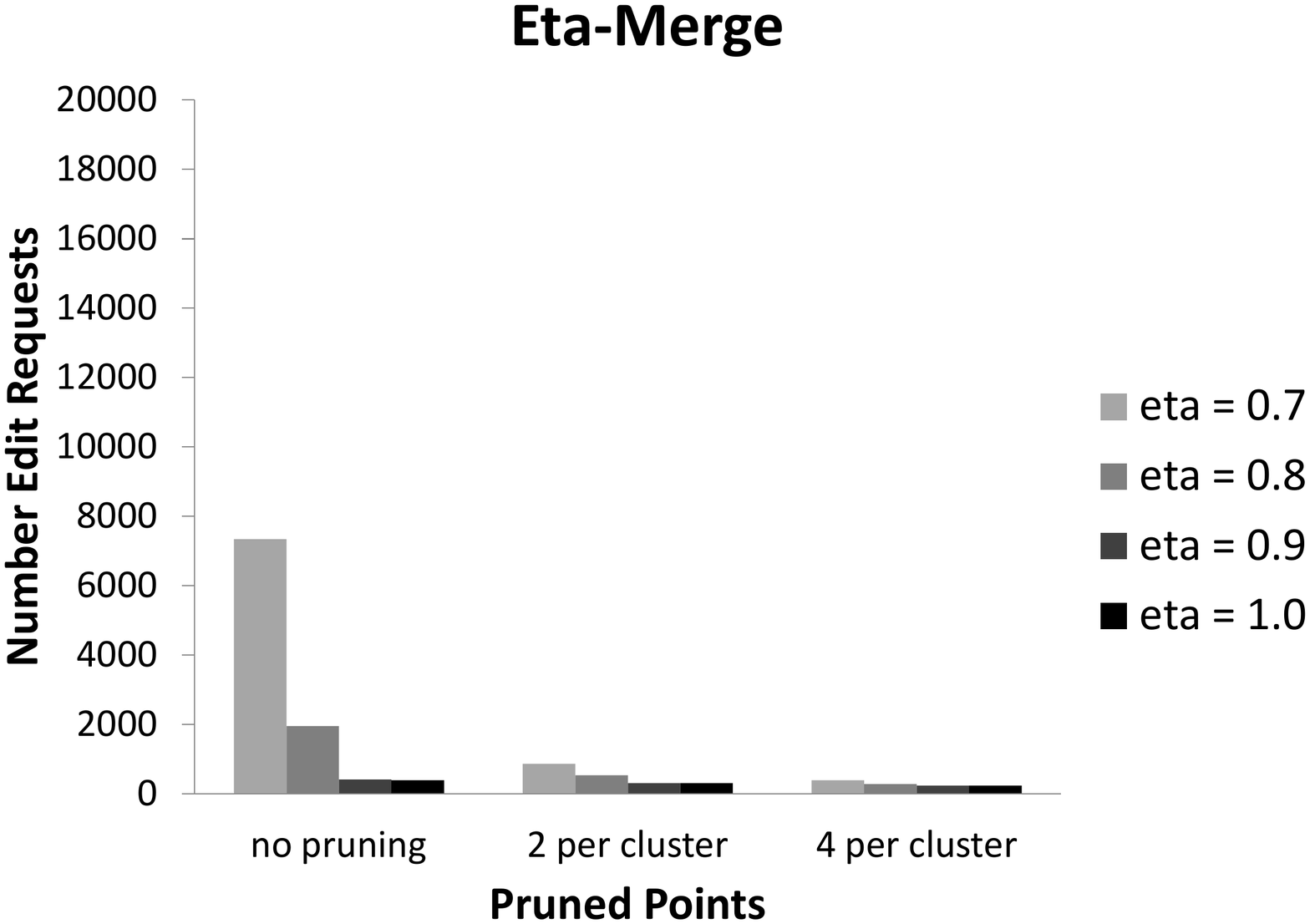}}
  \caption{Results in the $\eta$-merge model for data set A. The second chart corresponds to algorithms for correlation clustering error.}
    \label{fig:experiments-eta-merge}
	\end{figure}

\subsubsection{Experiments in the $\eta$-merge model}
We first experiment with local clustering algorithms in the $\eta$-restricted
merge setting. Here we use the algorithm in Figure~\ref{fig:split-average-linkage} to perform the splits, and the algorithm in Figure~\ref{fig:merge-average-linkage-relaxed} to perform the merges. We show the results of running our algorithm on data set A in Figure~\ref{fig:experiments-eta-merge}.  The complete experimental results are in the Apppendix.  We find that for larger settings of $\eta$, the number of edit requests (necessary to find the target clustering) is very favorable and is consistent with our theoretical analysis.  The results are better for pruned datasets, where we get very good performance regardless of the setting of $\eta$.  The results for algorithms in Figure~\ref{fig:split-average-linkage} and Figure~\ref{fig:merge-average-linkage-cc} (for the correlation-clustering objective) are very favorable as well.

\subsubsection{Experiments in the unrestricted-merge model}
We also experiment with algorithms in the unrestricted merge model. Here we
use the same algorithm to perform the splits, but use the algorithm in
Figure~\ref{fig:merge-average-linkage-unrestricted} to perform the merges.
We show the results on dataset A in Figure~\ref{fig:experiments-unrestricted-merge}.  The complete experimental results are in the Apppendix.  We find that for larger settings of $\eta$ our results are better than our theoretic analysis (we only show results for $\eta \ge 0.5$), and performance improves further for pruned datasets.  Our investigations show that for unpruned datasets and smaller settings of $\eta$, we are still able to quickly get close to the target clustering, but the algorithms are not able to converge to the target due to inconsistencies in the average-linkage tree.  We can address some of these inconsistencies by constructing the tree in a more robust way, which indeed gives improved performance for unpruned data sets.

\begin{figure}
\centering
\subfigure{\includegraphics[width=60mm,height=40mm]{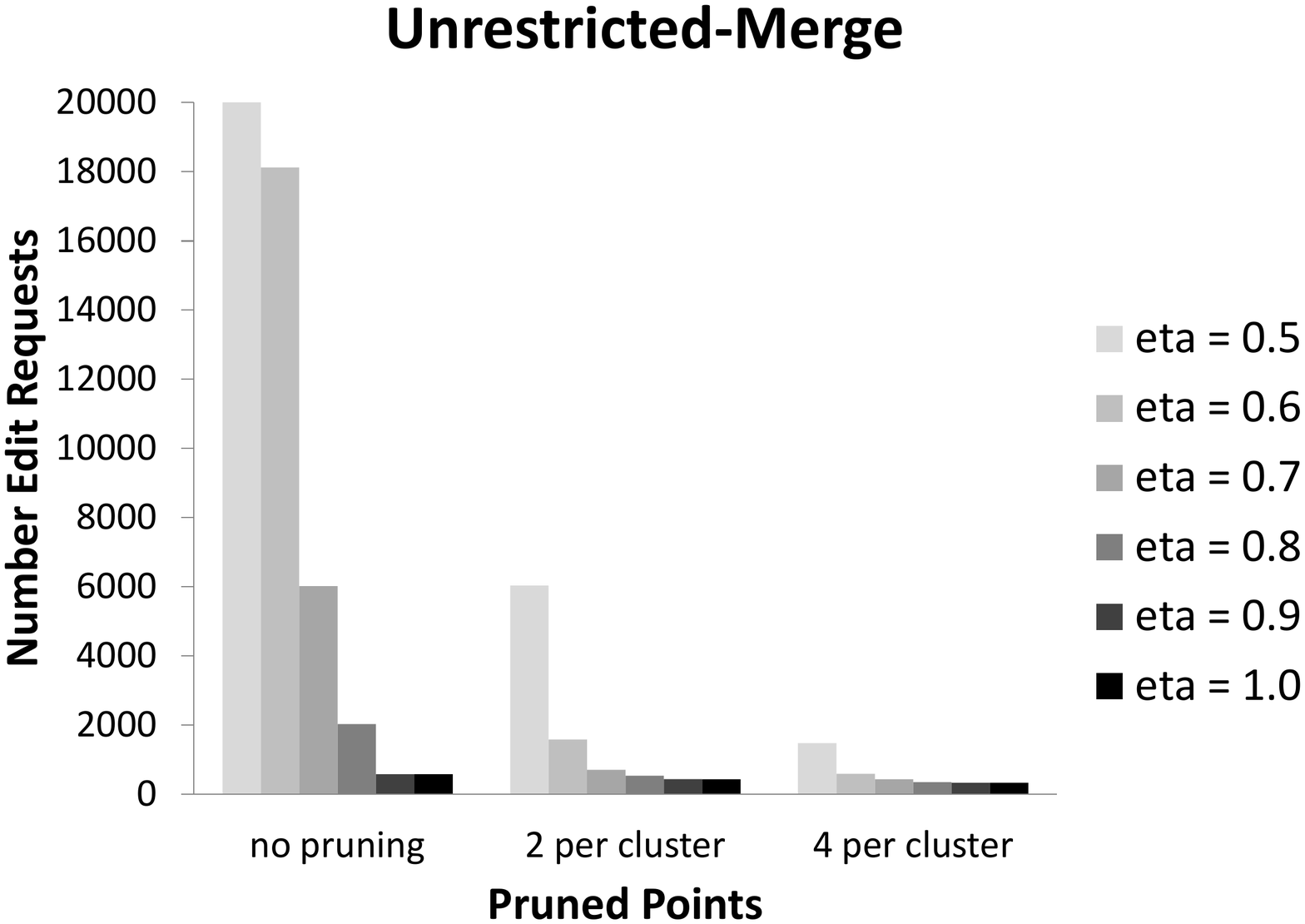}}
  \caption{Results in the unrestricted merge model for data set A.}
  \label{fig:experiments-unrestricted-merge}
  \end{figure}

\subsubsection{Experiments with small initial error}
{
We also consider a setting where the initial clustering is already very accurate.  In order to simulate this scenario, when we compute the initial clustering, for each document we keep its ground-truth cluster assignment with probability $0.95$, and otherwise reassign it to one of the other clusters, which is chosen uniformly at random.  This procedure usually gives us initial clusterings with over-clustering and under-clustering error between 5 and 20, and correlation-clustering error between 500 and 1000.  As expected, in this setting our interactive algorithms perform much better, especially on pruned data sets.  Figure~\ref{fig:small-error-experiments} displays the results; we can see that in these cases it often takes less than one hundred edit requests to find the target clustering in both models.}
\begin{figure}[!t]
\centering
\subfigure{\label{fig:small-starting-error-first-model}\includegraphics[width=60mm,height=50mm]{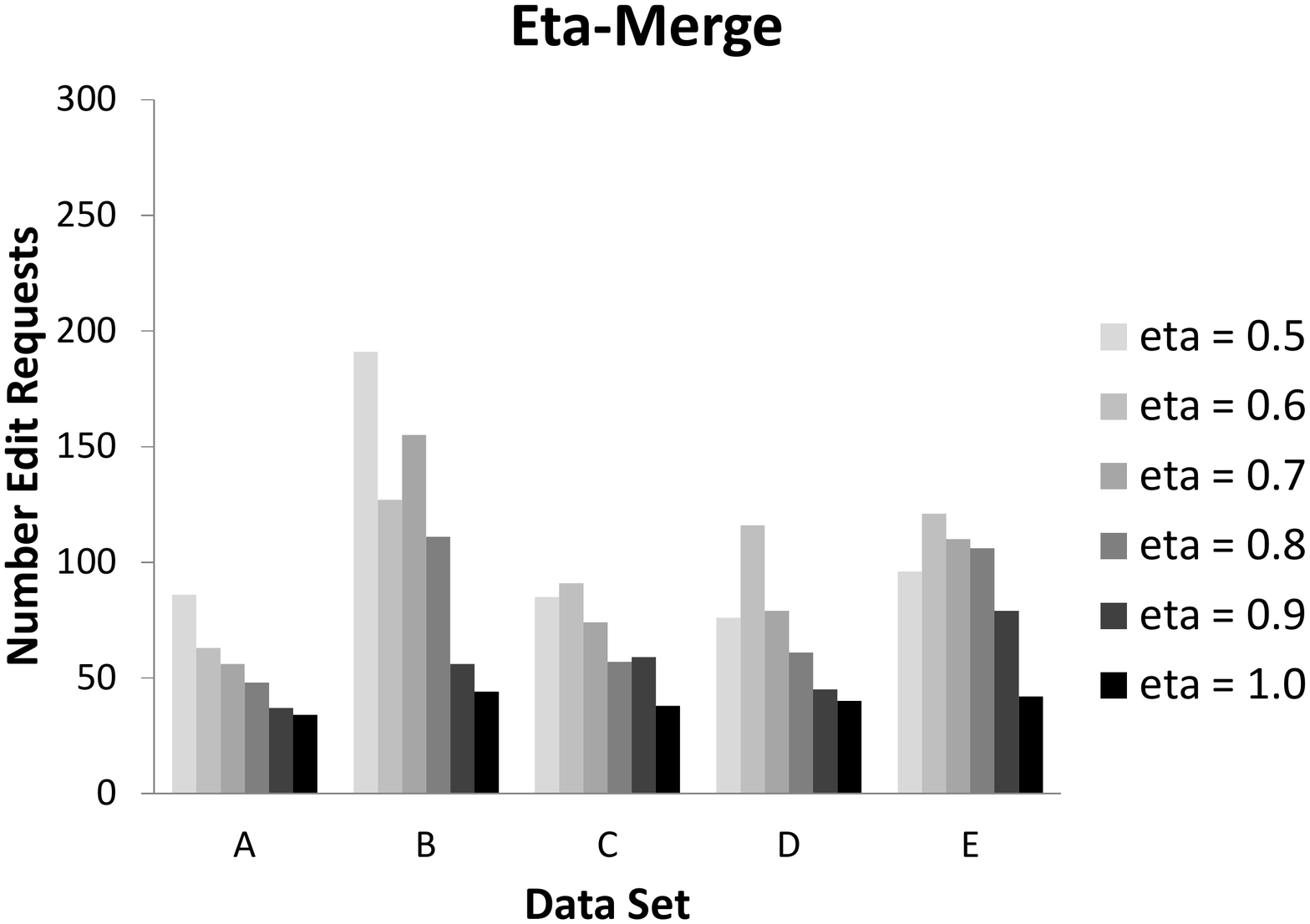}}
\subfigure{\label{fig:small-starting-error-cc-model}\includegraphics[width=60mm,height=50mm]{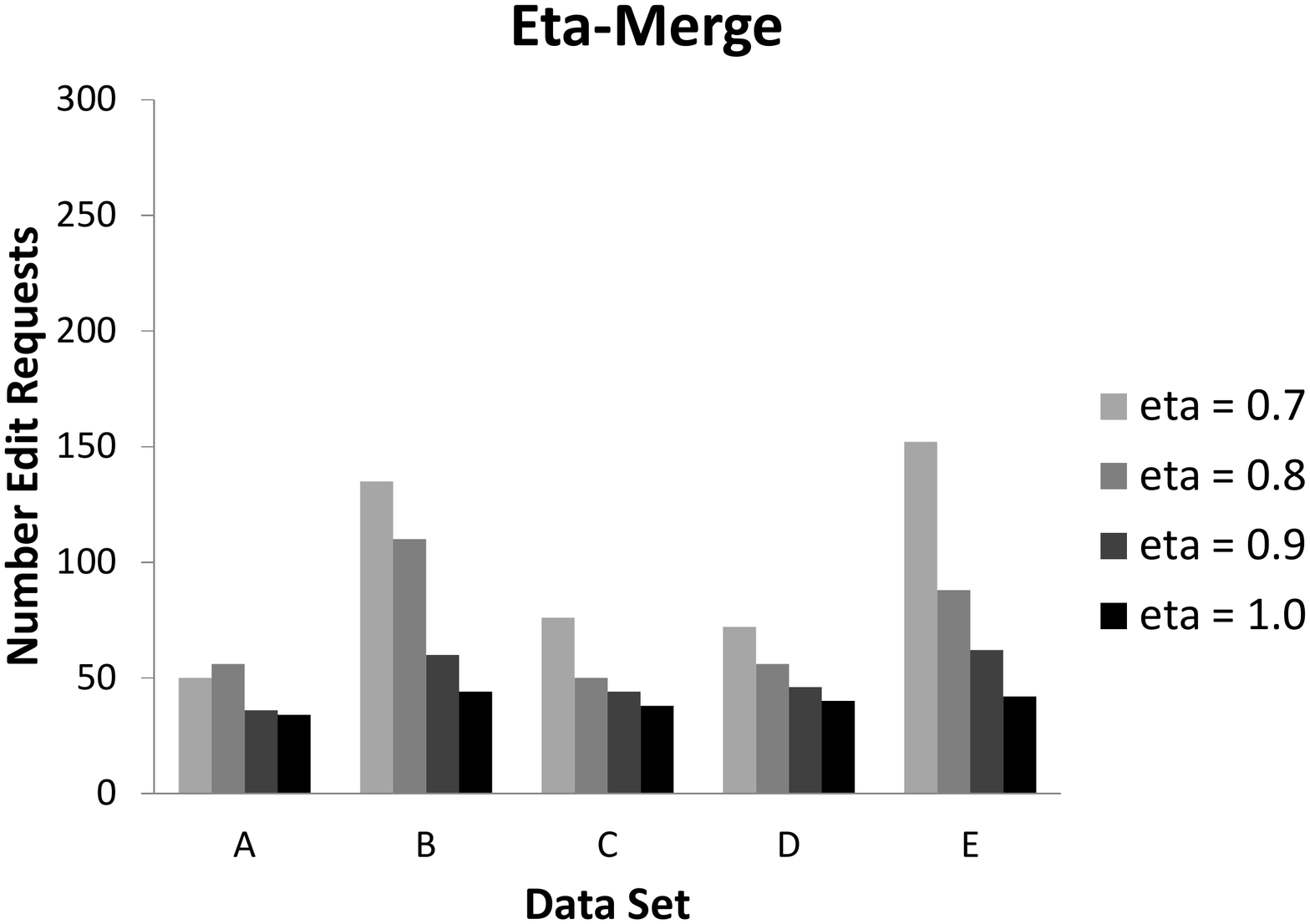}}
\subfigure{\label{fig:small-starting-error-second-model}\includegraphics[width=60mm,height=50mm]{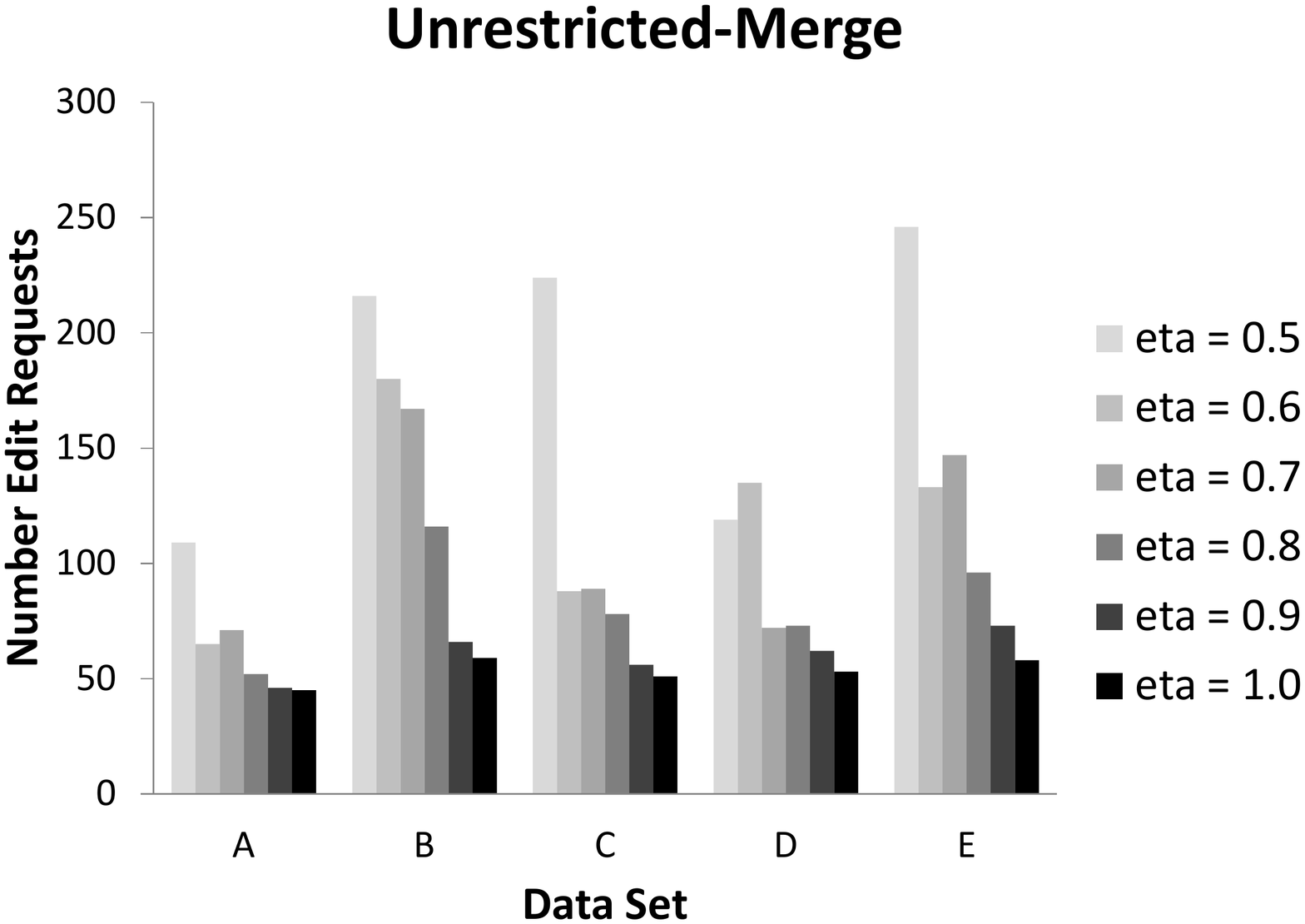}}
  \caption{Results for initial clusterings with small error. Results presented for pruned data sets (4 points per cluster). The second chart corresponds to algorithms for correlation clustering error.}
  \label{fig:small-error-experiments}
\end{figure}

\subsubsection{Improved performance using a robust average-linkage tree}

When we investigate the inconsistencies in the average linkage trees, we observe that there are ``outlier'' points that are attached near the root of the tree, which are incorrectly split off and re-merged by the algorithm without making any progress towards finding the target clustering.

We can address these outliers by constructing the average-linkage tree in a more robust way: first find groups (``blobs'') of similar points of some minimum size, compute an average-linkage tree for each group, and then merge these trees using average-linkage.  The tree constructed in such fashion may then be used by our algorithms.

We tried this approach, using Algorithm 2 from \cite{BG10} to compute the ``blobs''.  We find that using the robust average-linkage tree gives better performance for the unpruned data sets, but gives no gains for the pruned data sets.  Figure~\ref{fig:tree-construction-experiments} displays the comparison for the five unpruned data sets.  For the pruned data sets, it's likely that the robust tree and the standard tree are very similar, which explains why there is little difference in performance (results not shown).
\begin{figure}[!t]
  \centering
\subfigure{\includegraphics[width=60mm,height=50mm]{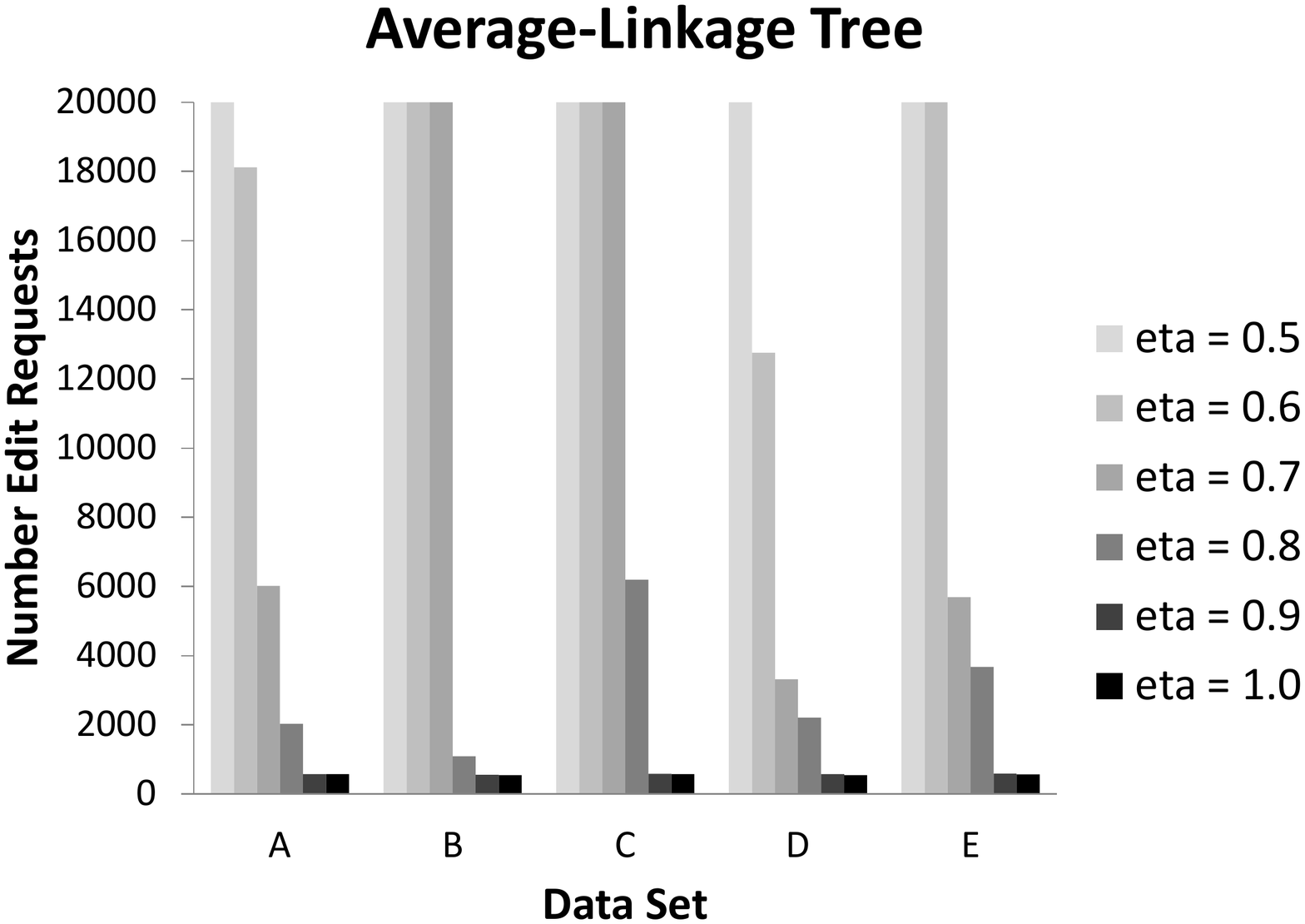}}
\subfigure{\includegraphics[width=60mm,height=50mm]{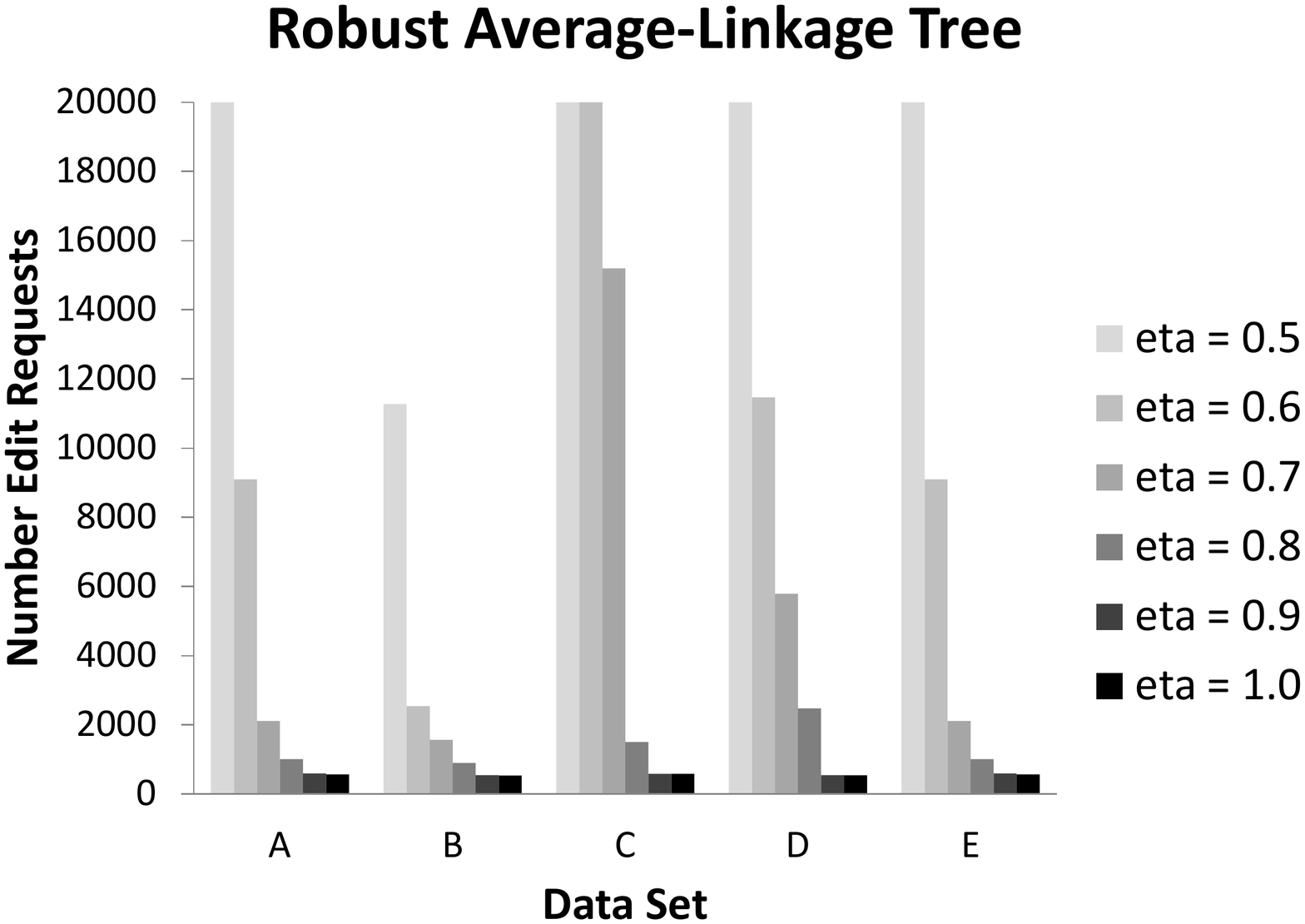}}
  \caption{Results in the unrestricted-merge model using a robust average-linkage tree.  Results presented for unpruned data sets.}
  \label{fig:tree-construction-experiments}
\end{figure}

\section{Discussion}

In this work we motivated and studied a new framework and algorithms for interactive clustering.  Our framework models practical constraints on the algorithms: we start with an initial clustering that we cannot modify arbitrarily, and are only allowed to make local edits consistent with user requests.  In this setting, we develop several simple, yet effective algorithms under different assumptions about the nature of the edit requests and the structure of the data.  We present theoretical analysis that shows that our algorithms converge to the target clustering after a small number of edit requests.  We also present experimental evidence that shows that our algorithms work well in practice.

Several directions come out of this work. It would be interesting to relax the condition on $\eta$ in the $\eta$-merge model, and the assumption about the request sequences in the unrestricted-merge model. It is important to study additional properties of an interactive clustering algorithm. In particular, it is often desirable that the algorithm never increase the error of the current clustering. Our algorithms in Figures~\ref{fig:split-average-linkage}, ~\ref{fig:merge-average-linkage-cc} and ~\ref{fig:merge-average-linkage-unrestricted} have this property, but the algorithm in Figure~\ref{fig:merge-average-linkage-relaxed} does not.

\bibliography{report}

\begin{thebibliography}{26}
\providecommand{\natexlab}[1]{#1}
\providecommand{\url}[1]{\texttt{#1}}
\expandafter\ifx\csname urlstyle\endcsname\relax
  \providecommand{\doi}[1]{doi: #1}\else
  \providecommand{\doi}{doi: \begingroup \urlstyle{rm}\Url}\fi

\bibitem[Achlioptas and McSherry(2005)]{AM05}
D.~Achlioptas and F.~McSherry.
\newblock On spectral learning of mixtures of distributions.
\newblock In \emph{Proceedings of the 18th Annual Conference on Learning
  Theory}, 2005.

\bibitem[Angluin(1998)]{angluin}
D.~Angluin.
\newblock Queries and concept learning.
\newblock \emph{Machine Learning}, 2:\penalty0 319--342, 1998.

\bibitem[Arora and Kannan(2001)]{AK01}
S.~Arora and R.~Kannan.
\newblock Learning mixtures of arbitrary {Gaussians}.
\newblock In \emph{Proceedings of the 33rd ACM Symposium on Theory of
  Computing}, 2001.

\bibitem[Awasthi and Zadeh(2010)]{AwasthiZ10}
Pranjal Awasthi and Reza~Bosagh Zadeh.
\newblock Supervised clustering.
\newblock In \emph{NIPS}, 2010.

\bibitem[Balcan and Blum(2008)]{BalcanB08}
Maria-Florina Balcan and Avrim Blum.
\newblock Clustering with interactive feedback.
\newblock In \emph{ALT}, 2008.

\bibitem[Balcan and Gupta(2010)]{BG10}
Maria-Florina Balcan and Pramod Gupta.
\newblock Robust hierarchical clustering.
\newblock In \emph{COLT}, 2010.

\bibitem[Balcan et~al.(2008)Balcan, Blum, and Vempala]{BalcanBV08}
Maria-Florina Balcan, Avrim Blum, and Santosh Vempala.
\newblock A discriminative framework for clustering via similarity functions.
\newblock In \emph{Proceedings of the 40th annual ACM symposium on Theory of
  computing}, STOC '08, 2008.

\bibitem[Bansal et~al.(2004)Bansal, Blum, and Chawla]{Bansal04}
Nikhil Bansal, Avrim Blum, and Shuchi Chawla.
\newblock Correlation clustering.
\newblock \emph{Machine Learning}, 56\penalty0 (1-3), 2004.

\bibitem[Basu et~al.(2004)Basu, Banjeree, Mooney, Banerjee, and Mooney]{Basu04}
Sugato Basu, A.~Banjeree, ER. Mooney, Arindam Banerjee, and Raymond~J. Mooney.
\newblock Active semi-supervision for pairwise constrained clustering.
\newblock In \emph{In Proceedings of the 2004 SIAM International Conference on
  Data Mining (SDM-04}, pages 333--344, 2004.

\bibitem[Belkin and Sinha(2010)]{BelkinS10}
Mikhail Belkin and Kaushik Sinha.
\newblock Polynomial learning of distribution families.
\newblock In \emph{FOCS}, 2010.

\bibitem[Boulis and Ostendorf(2004)]{Boulis04}
Constantinos Boulis and Mari Ostendorf.
\newblock Combining multiple clustering systems.
\newblock In \emph{In 8th European conference on Principles and Practice of
  Knowledge Discovery in Databases(PKDD), LNAI 3202}, 2004.

\bibitem[Brubaker and Vempala(2008)]{BV08}
S.~Charles Brubaker and Santosh Vempala.
\newblock Isotropic {PCA} and affine-invariant clustering.
\newblock \emph{CoRR}, abs/0804.3575, 2008.

\bibitem[Bryant and Berry(2001)]{Bryant01}
David Bryant and Vincent Berry.
\newblock A structured family of clustering and tree construction methods.
\newblock \emph{Adv. Appl. Math.}, 27\penalty0 (4), November 2001.

\bibitem[Dai et~al.(2010)Dai, Hu, and Niu]{Dai10}
Bo~Dai, Baogang Hu, and Gang Niu.
\newblock Bayesian maximum margin clustering.
\newblock In \emph{Proceedings of the 2010 IEEE International Conference on
  Data Mining}, ICDM '10, 2010.

\bibitem[Dasgupta(1999)]{Sanjoy99}
S.~Dasgupta.
\newblock Learning mixtures of {Gaussians}.
\newblock In \emph{Proceedings of the 40th Annual Symposium on Foundations of
  Computer Science}, 1999.

\bibitem[Dasgupta and Hsu(2008)]{Dasgupta08}
Sanjoy Dasgupta and Daniel Hsu.
\newblock Hierarchical sampling for active learning.
\newblock In \emph{ICML}, 2008.

\bibitem[Heller and Ghahramani(2005)]{HellerG05}
Katherine~A. Heller and Zoubin Ghahramani.
\newblock Bayesian hierarchical clustering.
\newblock In \emph{ICML}, 2005.

\bibitem[Kalai et~al.(2010)Kalai, Moitra, and Valiant]{KalaiMV10}
Adam~Tauman Kalai, Ankur Moitra, and Gregory Valiant.
\newblock Efficiently learning mixtures of two {Gaussians}.
\newblock In \emph{STOC}, 2010.

\bibitem[Kannan et~al.(2005)Kannan, Salmasian, and Vempala]{KSV05}
R.~Kannan, H.~Salmasian, and S.~Vempala.
\newblock The spectral method for general mixture models.
\newblock In \emph{Proceedings of the 18th Annual Conference on Learning
  Theory}, 2005.

\bibitem[Krishnamurthy et~al.(2012)Krishnamurthy, Balakrishnan, Xu, and
  Singh]{Krishnamurthy11}
Akshay Krishnamurthy, Sivaraman Balakrishnan, Min Xu, and Aarti Singh.
\newblock Efficient active algorithms for hierarchical clustering.
\newblock \emph{ICML}, 2012.

\bibitem[Meil{\u{a}}(2007)]{meila07}
Marina Meil{\u{a}}.
\newblock Comparing clusterings - an information based distance.
\newblock \emph{Journal of Multivariate Analysis}, 98\penalty0 (5):\penalty0
  873--895, 2007.

\bibitem[Moitra and Valiant(2010)]{MoitraV10}
Ankur Moitra and Gregory Valiant.
\newblock Settling the polynomial learnability of mixtures of gaussians.
\newblock In \emph{FOCS}, 2010.

\bibitem[Salton and Buckley(1988)]{SB88}
Gerard Salton and Christopher Buckley.
\newblock Term-weighting approaches in automatic text retrieval.
\newblock \emph{Information processing and management}, 24\penalty0
  (5):\penalty0 513--523, 1988.

\bibitem[Telgarsky and Dasgupta(2012)]{Telgarsky12}
Matus Telgarsky and Sanjoy Dasgupta.
\newblock Agglomerative {Bregman} clustering.
\newblock \emph{ICML}, 2012.

\bibitem[Voevodski et~al.(2012)Voevodski, Balcan, R\"{o}glin, Teng, and
  Xia]{Voevodski11}
Konstantin Voevodski, Maria-Florina Balcan, Heiko R\"{o}glin, Shang-Hua Teng,
  and Yu~Xia.
\newblock Active clustering of biological sequences.
\newblock \emph{Journal of Machine Learning Research}, 13:\penalty0 203--225,
  2012.

\bibitem[Zhong(2005)]{Zhong05}
Shi Zhong.
\newblock Generative model-based document clustering: a comparative study.
\newblock \emph{Knowledge and Information Systems}, 2005.

\end{thebibliography}

\appendix

\section{Complete Experimental Results}

The following figures show the complete experimental results for all the algorithms.  Figure~\ref{fig:experimental-results-1-app} and Figure~\ref{fig:experimental-results-2-app} give the results in the $\eta$-merge model.   Figure~\ref{fig:experimental-results-cc-1-app} and Figure~\ref{fig:experimental-results-cc-2-app} give the results in the $\eta$-merge model for the algorithms in Figure~\ref{fig:split-average-linkage} and Figure~\ref{fig:merge-average-linkage-cc} (for the correlation-clustering objective).  Figure~\ref{fig:experimental-results-unrestricted-1-app} and Figure~\ref{fig:experimental-results-unrestricted-2-app} give the results in the unrestricted-merge model. 

\begin{figure}[htbp]
  \centering
  \subfigure{\includegraphics[width=50mm,height=40mm]{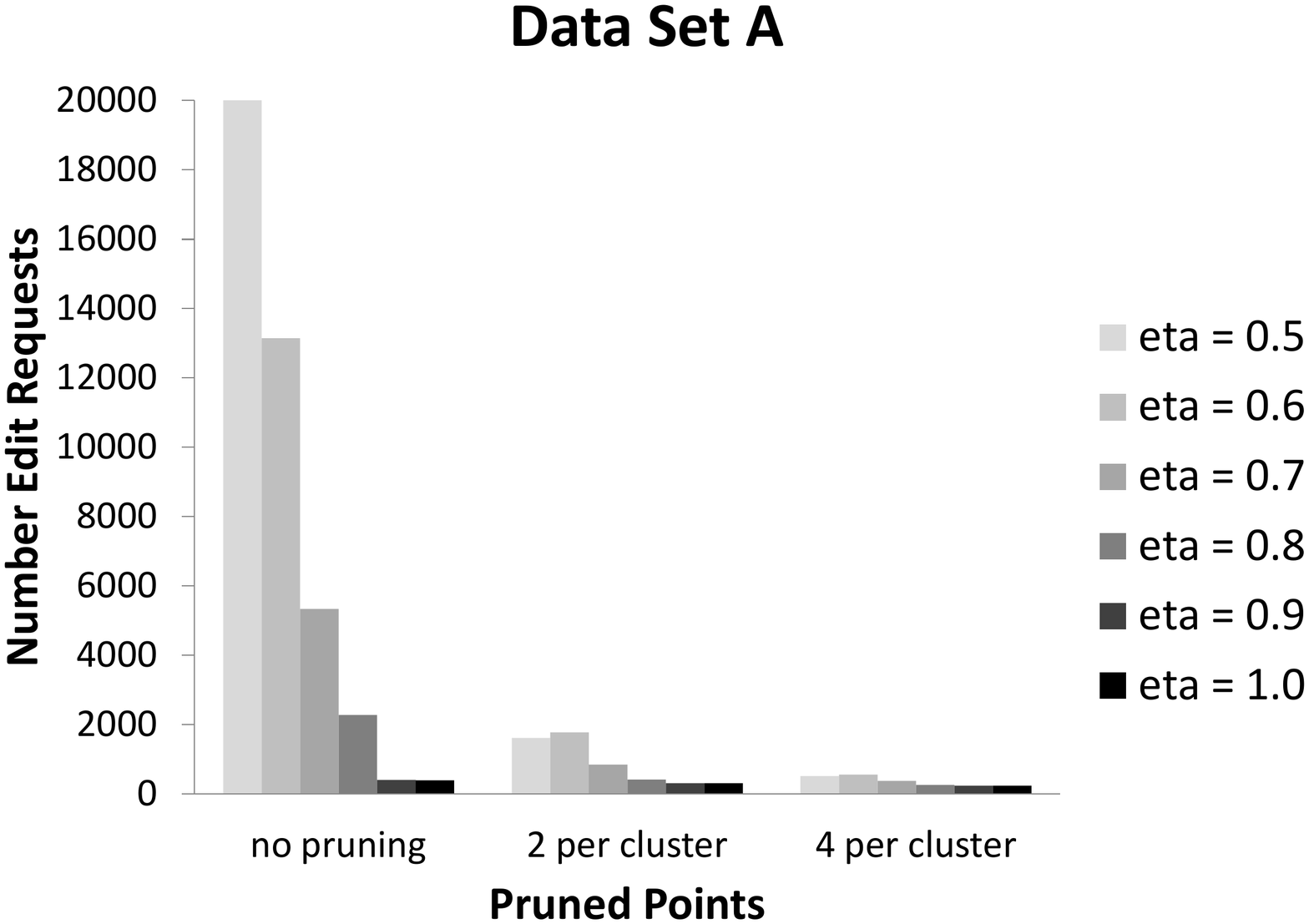}}
  \subfigure{\includegraphics[width=50mm,height=40mm]{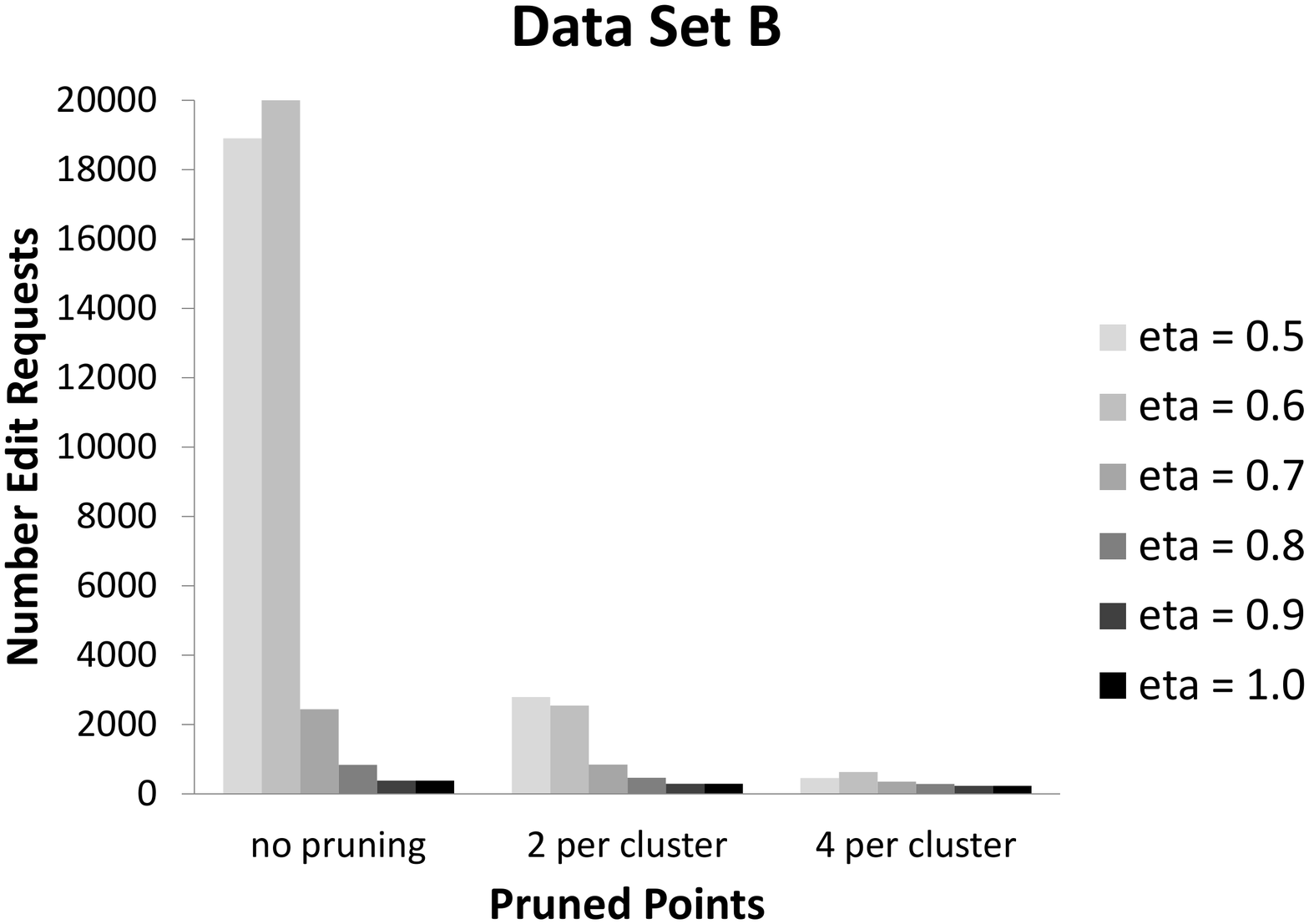}}
\subfigure{\includegraphics[width=50mm,height=40mm]{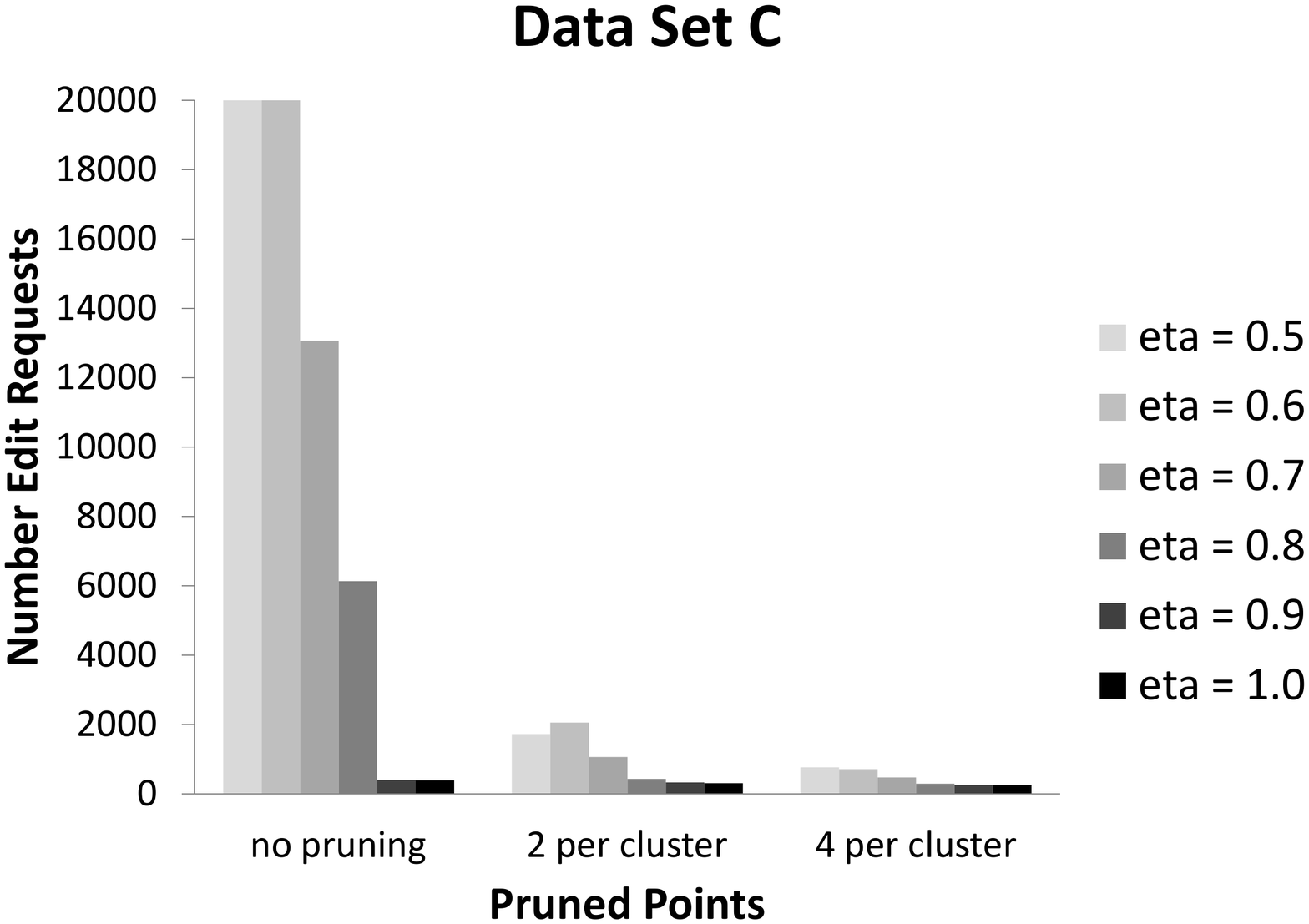}}
\caption{Results in the $\eta$-merge model on datasets A, B and C.}
\label{fig:experimental-results-1-app}
\end{figure}
\begin{figure}[htbp]
\centering
\subfigure{\includegraphics[width=50mm,height=40mm]{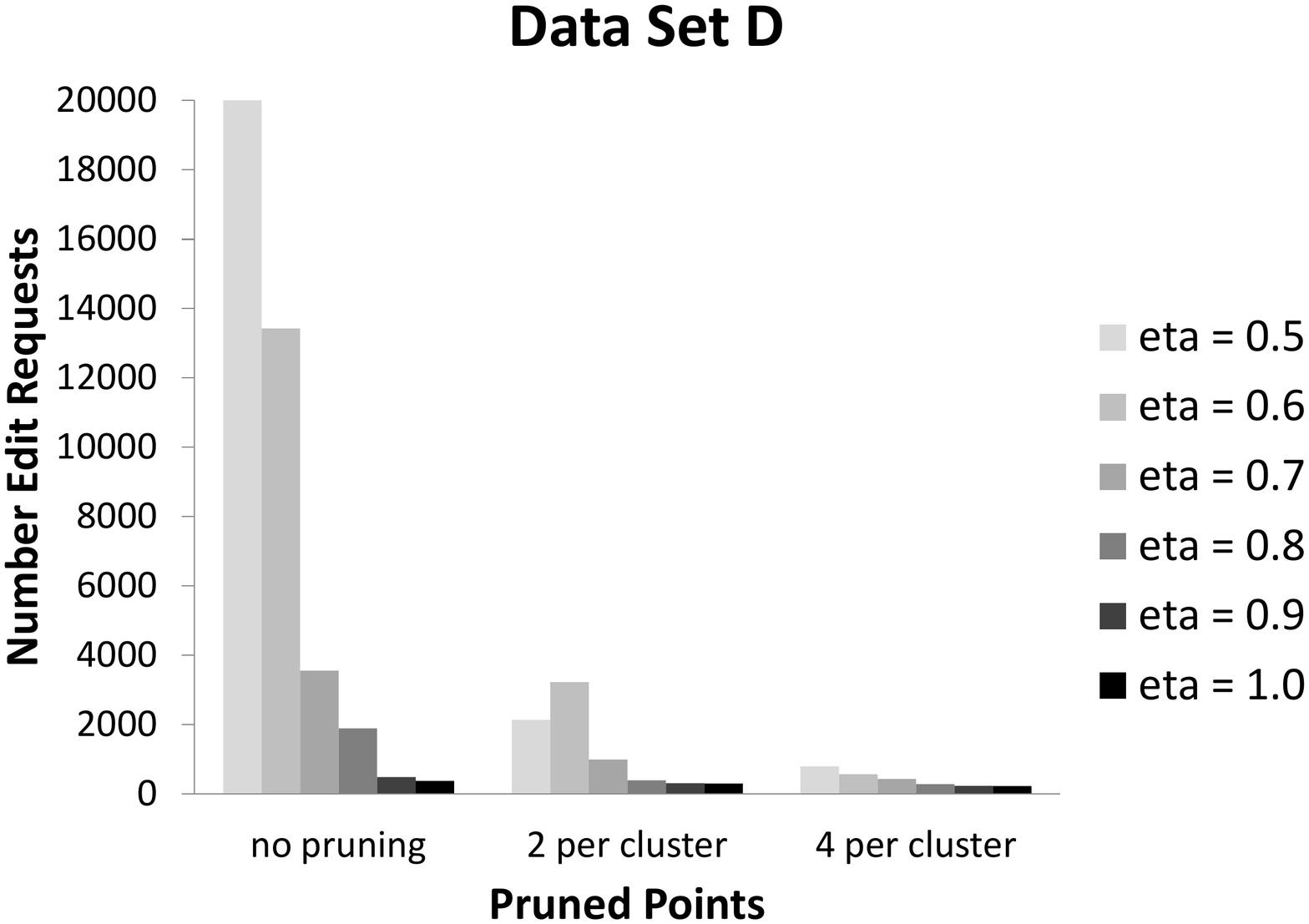}}
\subfigure{\includegraphics[width=50mm,height=40mm]{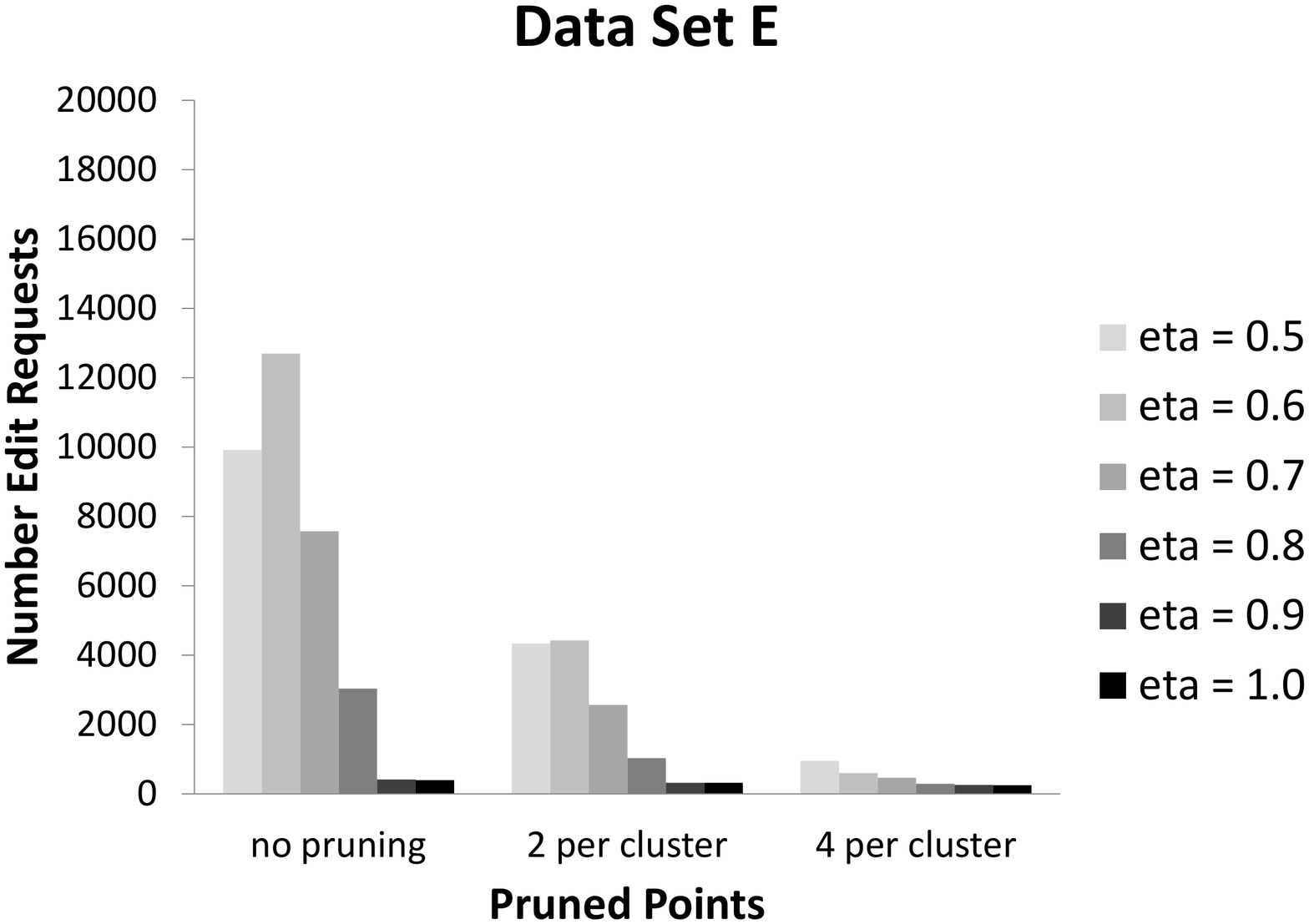}}
  \caption{Results in the $\eta$-merge model on datasets D and E.}
  \label{fig:experimental-results-2-app}
\end{figure}
\begin{figure}
\centering
 \subfigure{\includegraphics[width=50mm,height=40mm]{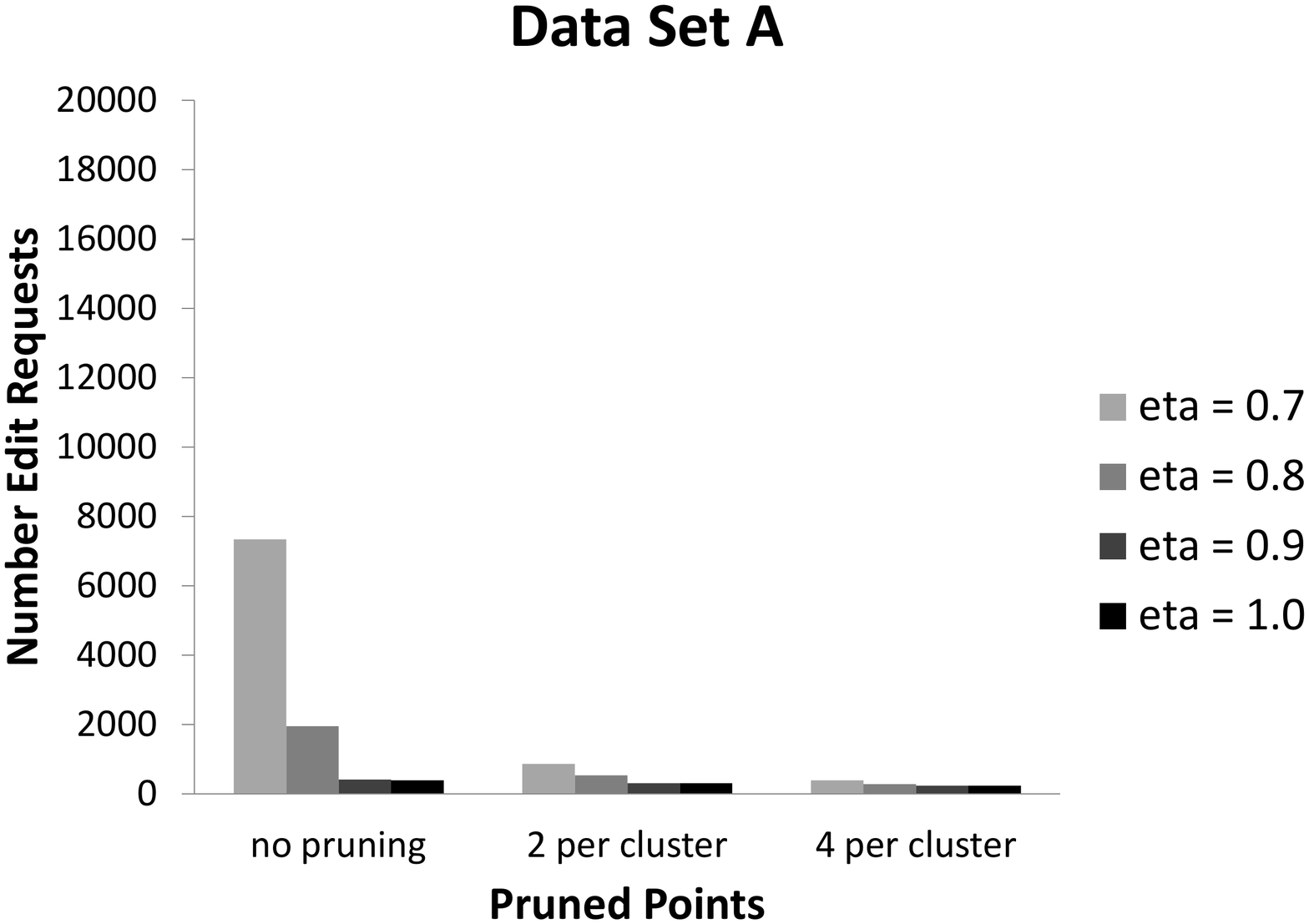}}
  \subfigure{\includegraphics[width=50mm,height=40mm]{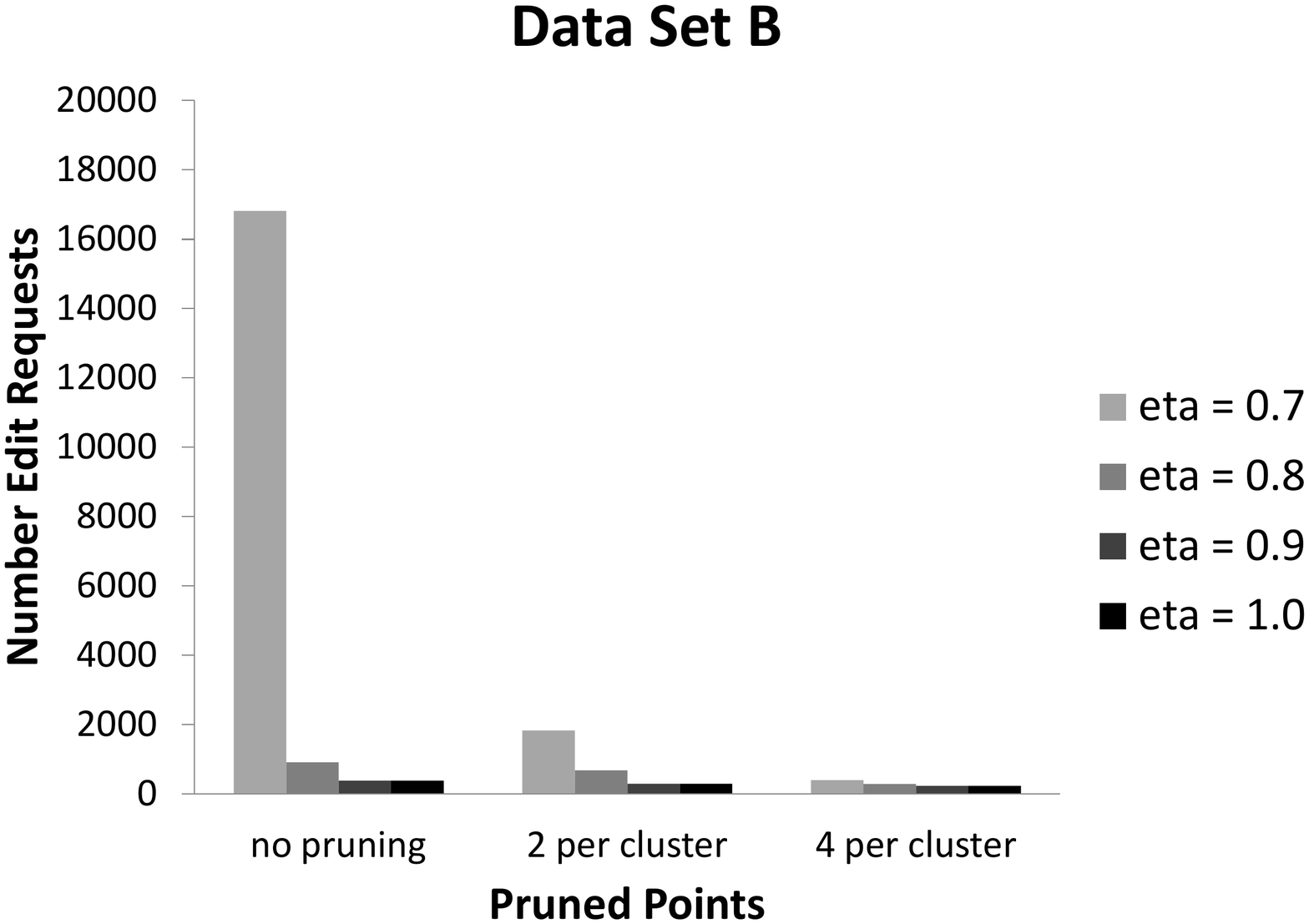}}
\subfigure{\includegraphics[width=50mm,height=40mm]{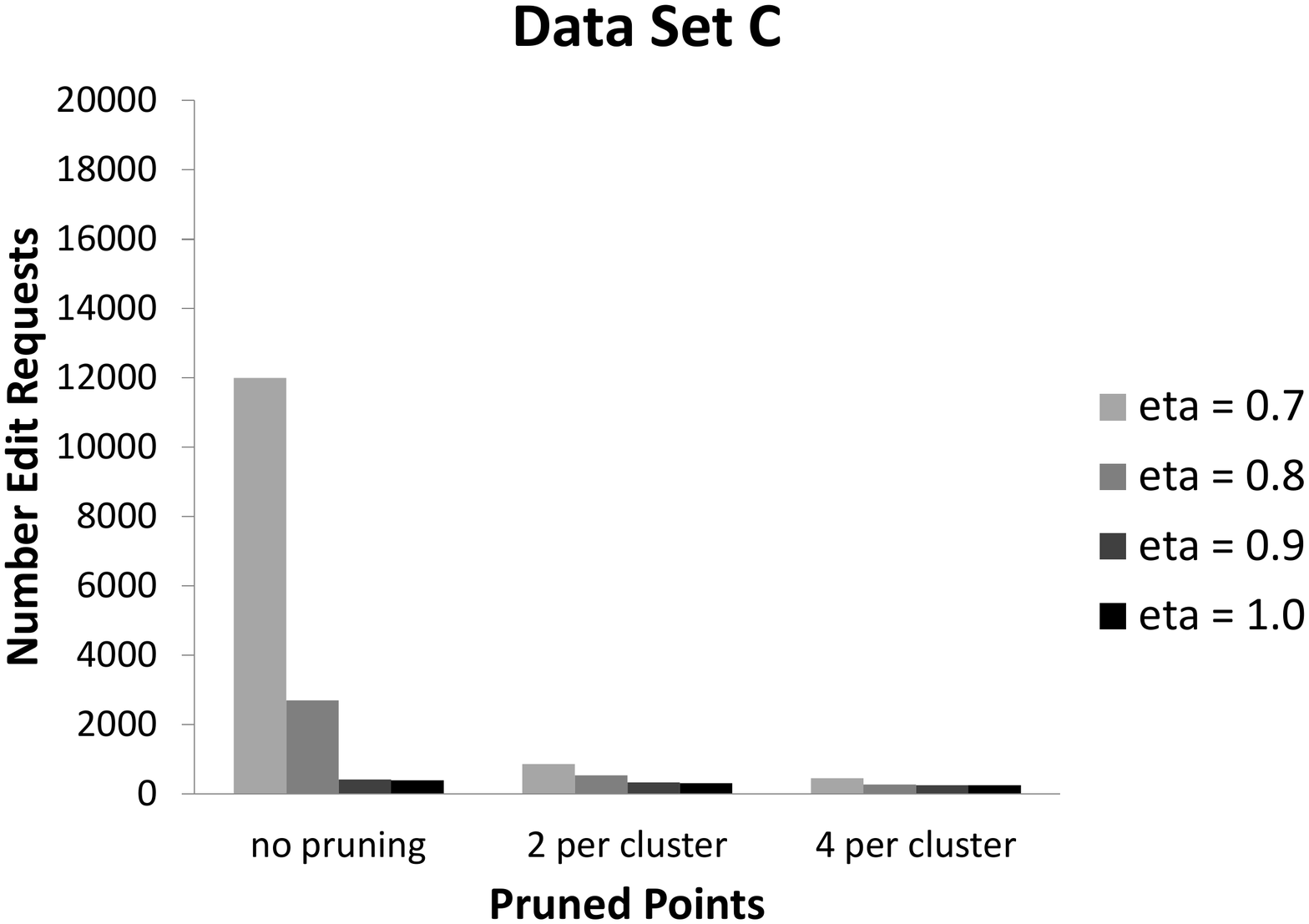}}
  \caption{Results in the $\eta$-merge model for algorithms for the correlation-clustering objective on datasets A, B and C.}
  \label{fig:experimental-results-cc-1-app}
\end{figure}
\begin{figure}
\centering
 \subfigure{\includegraphics[width=50mm,height=40mm]{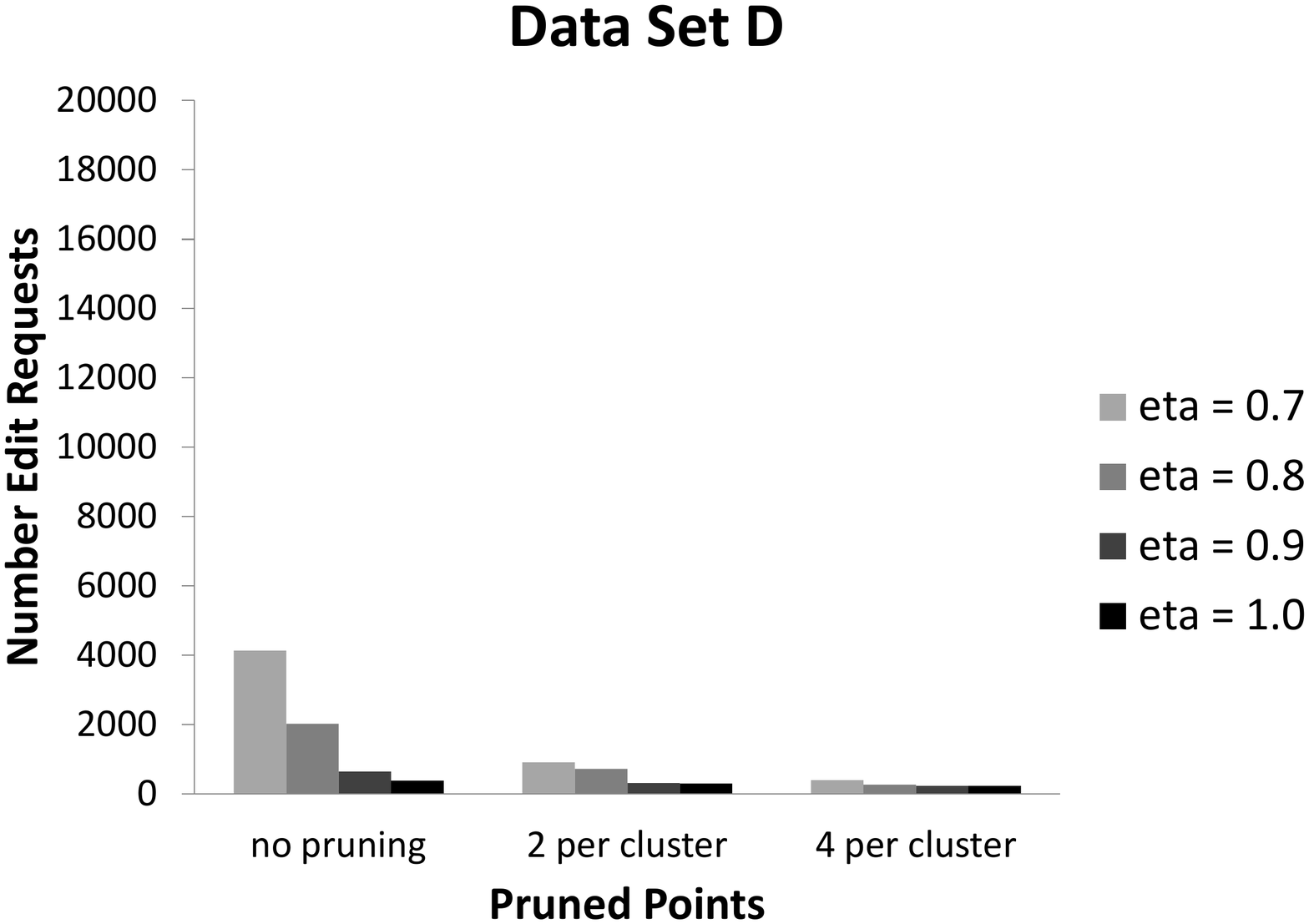}}
  \subfigure{\includegraphics[width=50mm,height=40mm]{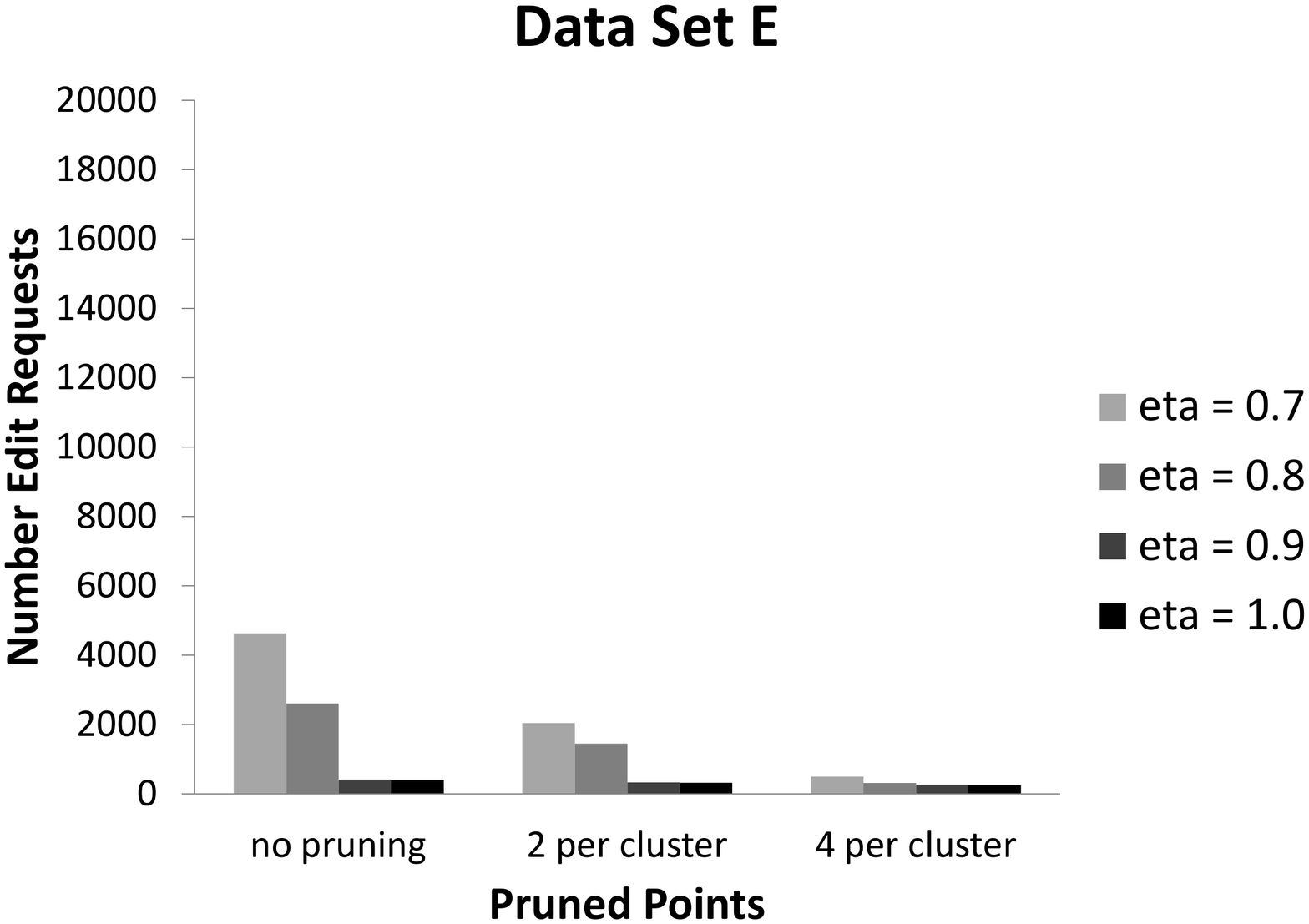}}
  \caption{Results in the $\eta$-merge model for algorithms for the correlation-clustering objective on datasets D and E.}
  \label{fig:experimental-results-cc-2-app}
\end{figure}
\begin{figure}[htbp]
  \centering
\subfigure{\includegraphics[width=50mm,height=40mm]{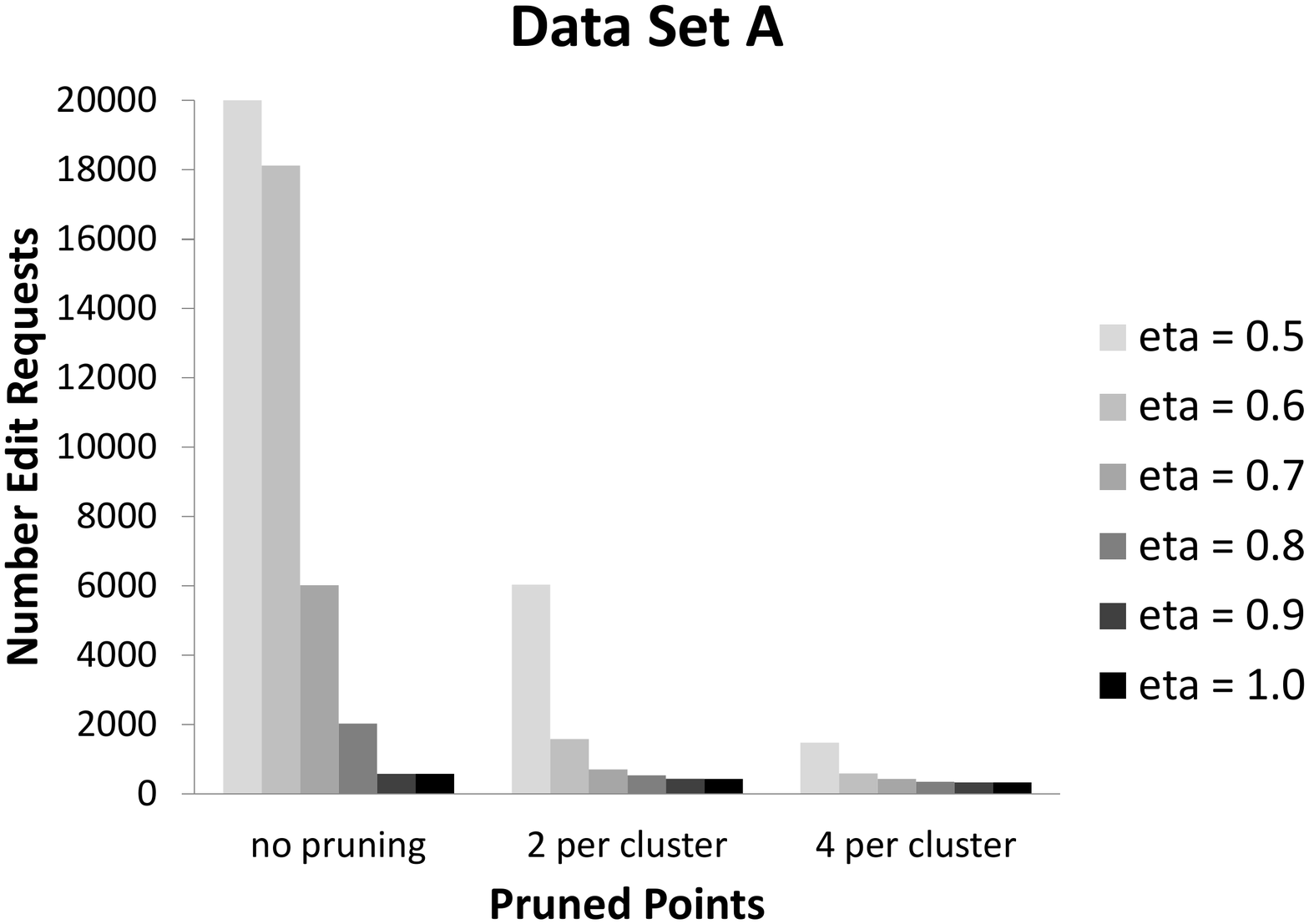}}
\subfigure{\includegraphics[width=50mm,height=40mm]{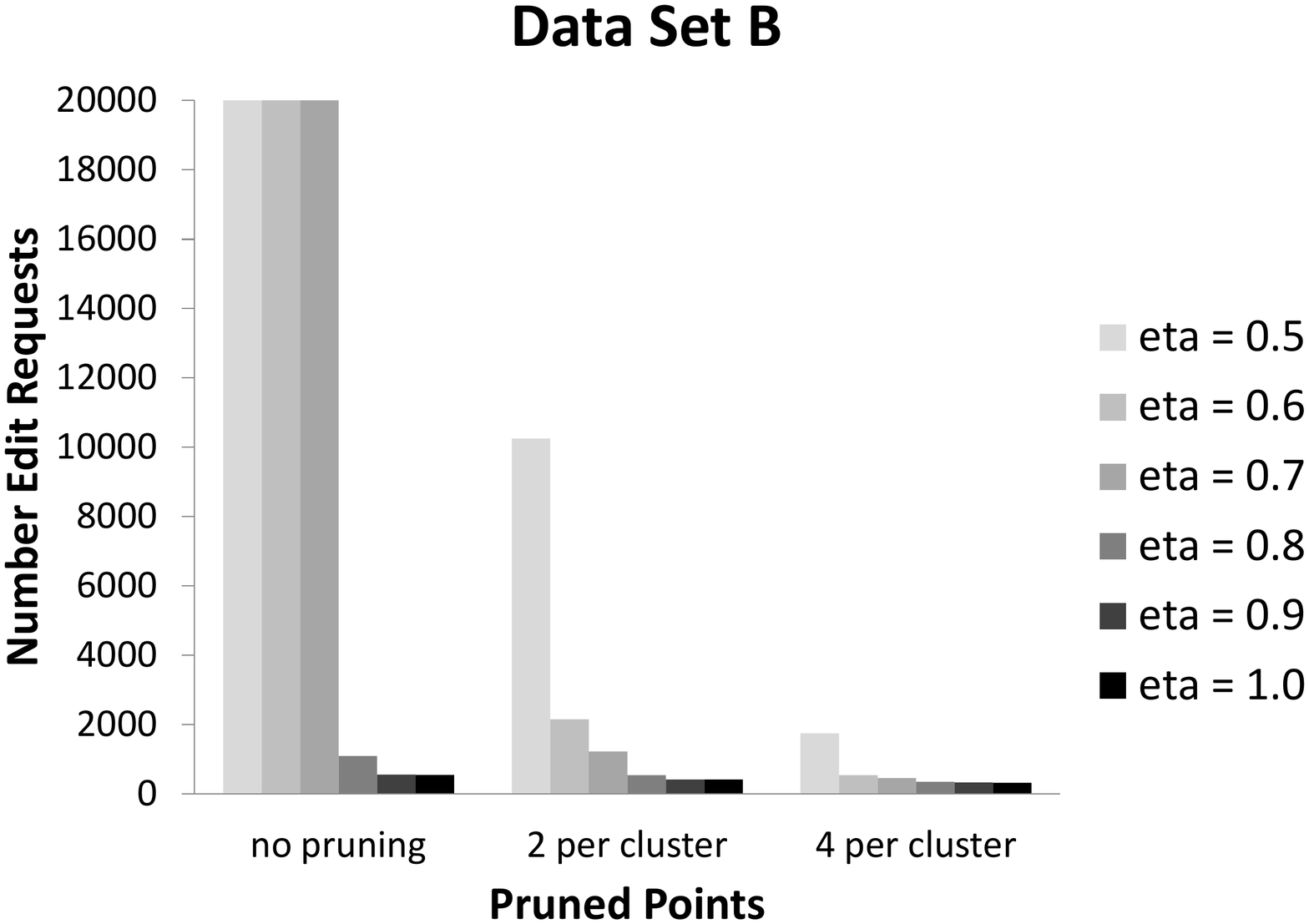}}
\subfigure{\includegraphics[width=50mm,height=40mm]{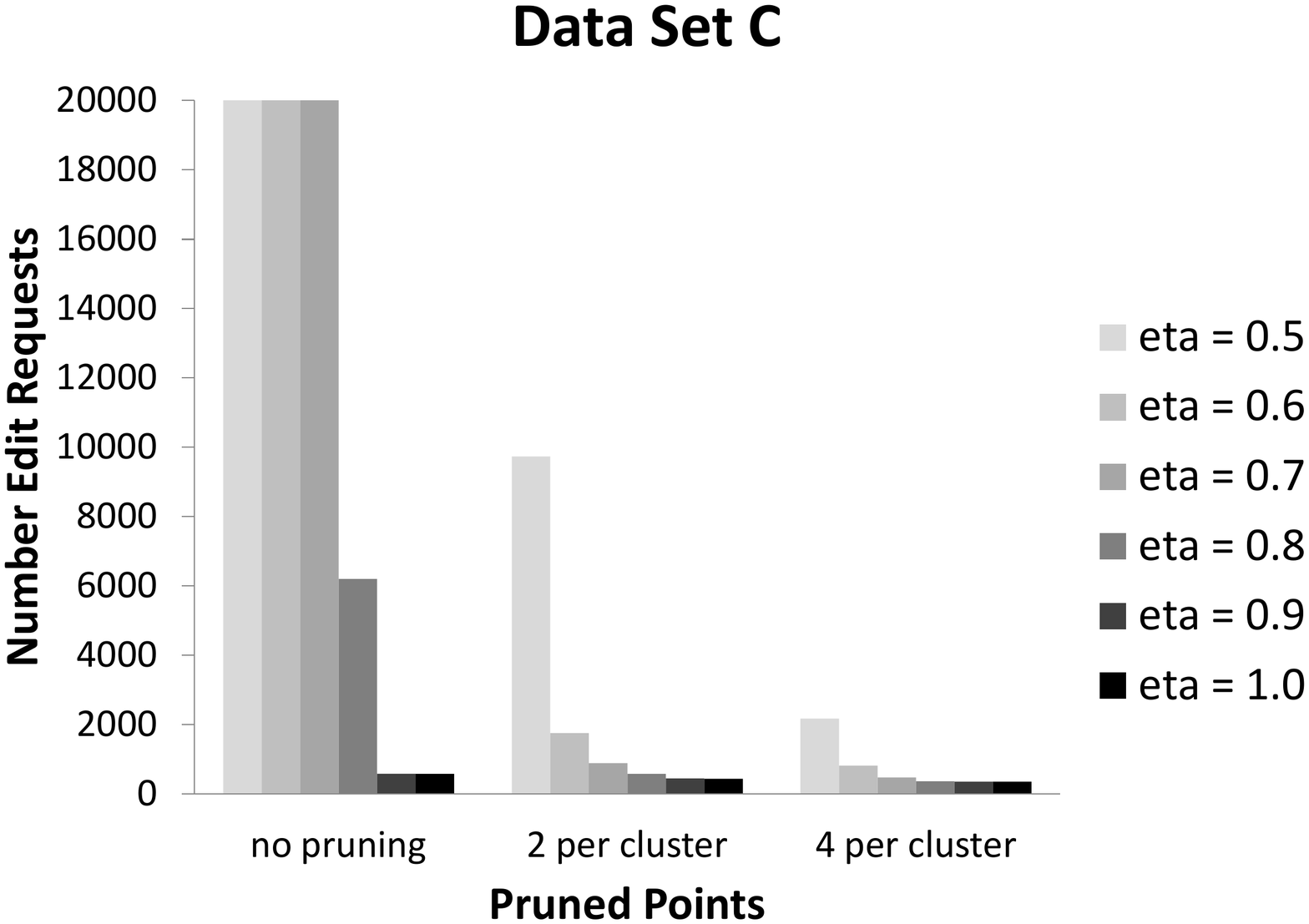}}
\caption{Results in the unrestricted-merge model on datasets A, B and C.}
  \label{fig:experimental-results-unrestricted-1-app}
\end{figure}
\begin{figure}
\centering
\subfigure{\includegraphics[width=50mm,height=40mm]{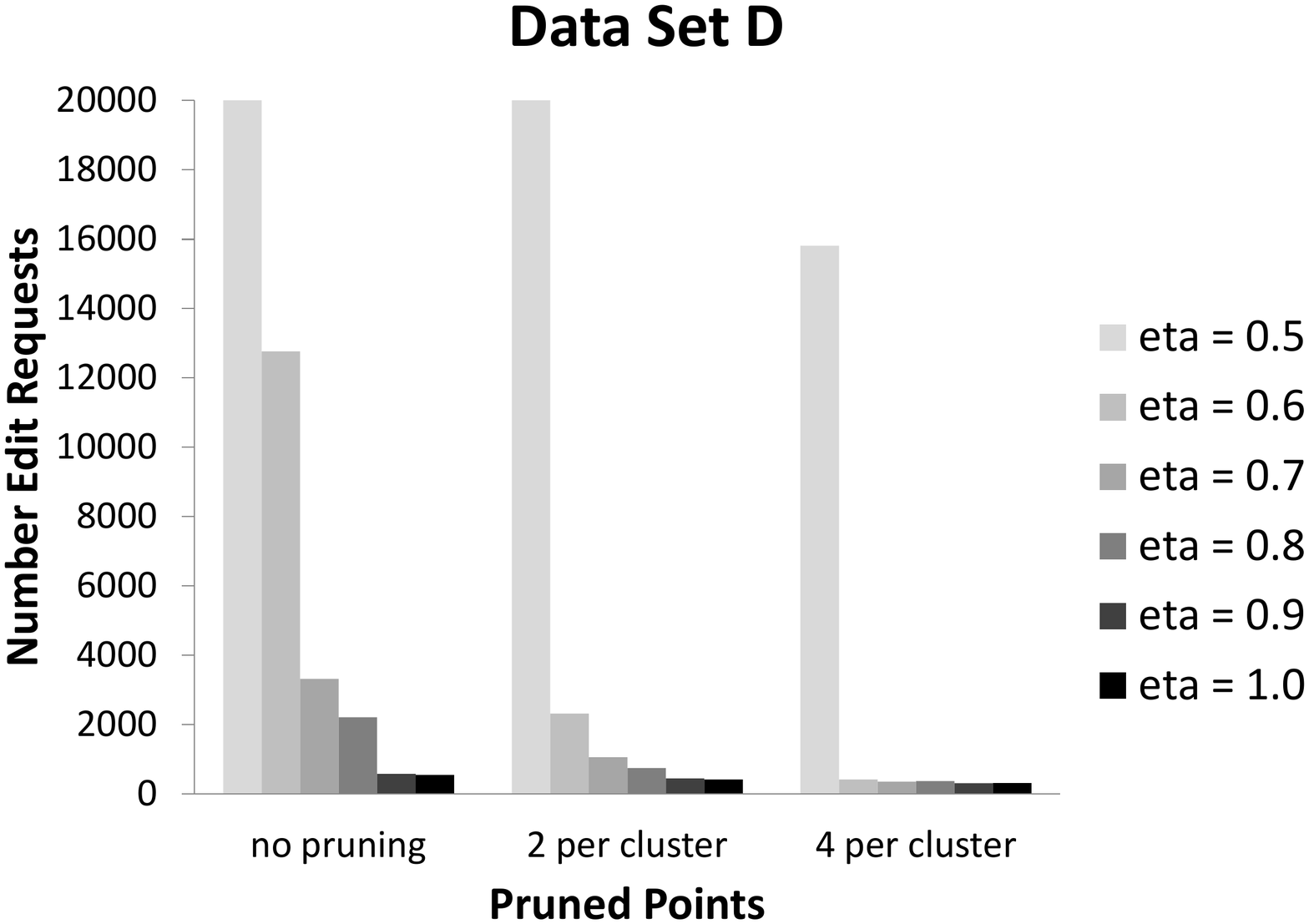}}
\subfigure{\includegraphics[width=50mm,height=40mm]{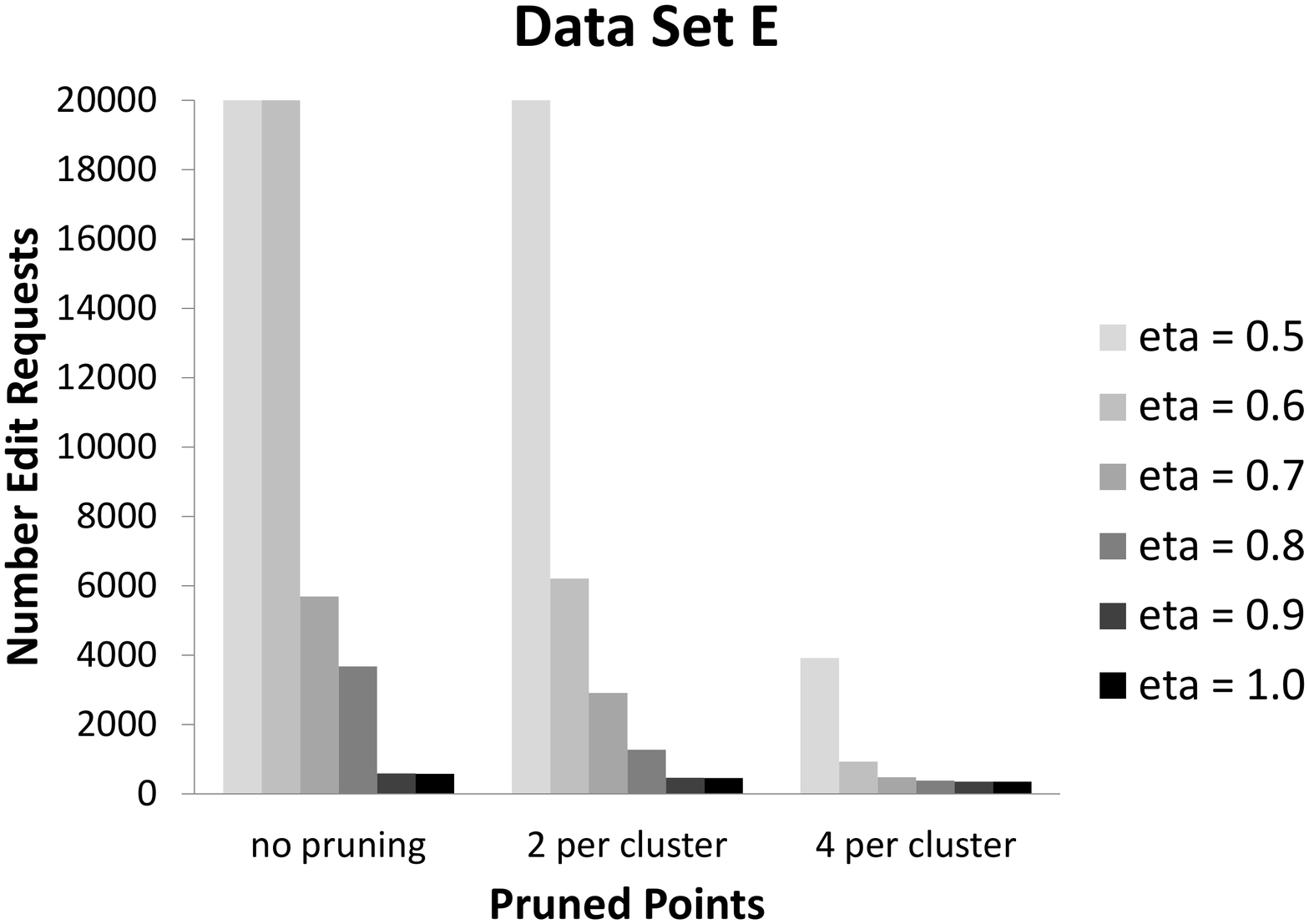}}
  \caption{Results in the unrestricted-merge model on datasets D and E.}
  \label{fig:experimental-results-unrestricted-2-app}
\end{figure}

\end{document}